\documentclass[11pt]{article}

\usepackage{epsf}
\usepackage{latexsym}
\usepackage{amssymb}
\usepackage{amsmath}
\usepackage[usenames,dvipsnames]{xcolor}
\usepackage[pagebackref,letterpaper=true,colorlinks=true,pdfpagemode=none,urlcolor=blue,linkcolor=blue,citecolor=BrickRed,pdfstartview=FitH]{hyperref}
\usepackage{prettyref}

\usepackage{stmaryrd}

\usepackage{verbatim}

\usepackage{fullpage}
\usepackage{amssymb}
\usepackage{amsmath}
\usepackage{amsopn}
\usepackage{graphics}
\usepackage{epsfig}
\usepackage{color}
\usepackage{enumerate}
\usepackage{cite}
\usepackage{graphicx}
\usepackage{version}
\usepackage[lined,boxed,commentsnumbered]{algorithm2e}
\usepackage{times}
\usepackage[bottom]{footmisc}

\newtheorem{theorem}{Theorem}[section]
\newtheorem{lemma}[theorem]{Lemma}
\newtheorem{definition}[theorem]{Definition}
\newtheorem{corollary}[theorem]{Corollary}
\newtheorem{claim}[theorem]{Claim}
\newtheorem{fact}[theorem]{Fact}

\newtheorem{remark}[theorem]{Remark}

\newcommand{\setdiff}{\Delta}

\newcommand{\NP}{\mbox{\rm NP}}

\newcommand{\val}{\mathsf{val}}
\newcommand{\Var}{\mathsf{Var}}
\DeclareMathOperator{\E}{E}

\def\calS{{\cal S}}
\def\calI{{\cal I}}

\def\hf{\hat{f}}
\def\eps{\epsilon}
\def\hh{\hat{h}}
\def\hg{\hat{g}}
\def\time_h{3}
\def\maxcut{{\sc MaxCut}}
\def\SSE{{\sc SmallSet Expansion}}
\def\maxthreesat{{\sc Max3SAT}}
\def\maxxor{{\sc MaxXOR}}
\def\maxthreexor{{\sc Max3XOR}}
\def\mincut{{\sc MinCUT}}
\def\maxbisect{{\sc MaxBisection}}
\def\UG{{\sc UniqueGames}}
\def\NP{{\sc NP}}

\newcommand{\sq}{\hbox{\rlap{$\sqcap$}$\sqcup$}}
\newcommand{\qed}{\hspace*{\fill}\sq}
\newenvironment{proof}{\noindent {\it Proof.}\ }{\qed\par\vskip 4mm\par}
\newenvironment{proofof}[1]{\bigskip \noindent {\it Proof of #1.}\quad }
{\qed\par\vskip 4mm\par}

\includeversion{version:full}
\excludeversion{version:10pages}

\begin{document}

\title{Parameterized Algorithms for Constraint Satisfaction Problems Above Average with Global Cardinality Constraints}
\author{Xue Chen\thanks{Supported by NSF Grant CCF-1526952.}
\\
Computer Science Department\\
University of Texas at Austin\\
{\tt xchen@cs.utexas.edu}
\and Yuan Zhou\\
Computer Science Department\\
Indiana University at Bloomington\\
{\tt yzhoucs@indiana.edu}
}

\maketitle

\begin{abstract}
Given a constraint satisfaction problem (CSP) on $n$ variables, $x_1, x_2, \dots, x_n \in \{\pm 1\}$, and $m$ constraints, a global cardinality constraint has the form of $\sum_{i = 1}^{n} x_i = (1-2p)n$, where $p \in (\Omega(1), 1 - \Omega(1))$ and $pn$ is an integer. Let $AVG$ be the expected number of constraints satisfied by randomly choosing an assignment to $x_1, x_2, \dots, x_n$, complying with the global cardinality constraint. The CSP above average with the global cardinality constraint problem asks whether there is an assignment (complying with the cardinality constraint) that satisfies more than $(AVG+t)$ constraints, where $t$ is an input parameter.

In this paper, we present an algorithm that finds a valid assignment satisfying more than $(AVG+t)$ constraints (if there exists one) in time $(2^{O(t^2)} + n^{O(d)})$. Therefore, the CSP above average with the global cardinality constraint problem is fixed-parameter tractable.
\end{abstract}
\thispagestyle{empty}
\setcounter{page}{0}
\pagebreak

\section{Introduction}
In a $d$-ary constraint satisfaction problem (CSP), we are given a set of boolean variables $\{x_1, x_2, \cdots,x_n\}$ over $\{\pm 1\}$ and $m$ constraints $C_1,\cdots,C_m$, where each constraint $C_i$ consists of a predicate on at most $d$ variables. A constraint is satisfied if and only if the assignment of the related variables is in the predicate of the constraint. The task is to find an assignment to $\{x_1,\cdots,x_n\}$ so that the greatest (or the least) number of constraints in $\{C_1,\cdots,C_m\}$ are satisfied. Simple examples of binary CSPs are the \maxcut~problem and the \mincut~problem, where each constraint includes 2 variables, and the predicate is always $\{(-1, +1), (+1, -1)\}$. The \maxthreesat~problem is an example of a ternary CSP. In a \maxthreesat~problem, each constraint includes 3 variables, and the predicate includes $7$ out $8$ possible assignments to the 3 variables. Another classical example of a ternary CSP is called the \maxthreexor~problem, where each predicate is either $\{(-1, -1, -1), (-1, +1, +1), (+1, -1, +1),$ $ (+1, +1, -1)\}$ or \\$\{(+1, +1, +1), (+1, -1, -1), (-1, +1, -1), (-1, -1, +1)\}$.

\paragraph{Constraint satisfaction problem above average.} For each CSP problem, there is a trivial randomized algorithm which chooses an assignment uniformly at random from all possible assignments. In a seminal work \cite{Hastad01}, H{\aa}stad showed that for \maxthreesat~and \maxthreexor,  it is \NP-hard to find an assignment satisfying $\epsilon m$ more constraints than the trivial randomized algorithm (in expectation) for an arbitrarily small constant $\epsilon > 0$. In other words, there is no non-trivial approximation algorithm (i.e. better approximation than the trivial randomization) for \maxthreesat~or \maxthreexor~assuming {{\sc P}}$\neq${\sc NP}.  Recently, this result was extended to many more CSPs by Chan \cite{Chan}.  Assuming the Unique Games Conjecture \cite{khot}, Austrin and Mossel \cite{AM09} provided a sufficient condition for CSPs to admit non-trivial approximation algorithms. Later, Khot, Tulsiani and Worah\cite{KTW} gave a complete characterization for these CSPs. Guruswami et al. \cite{GHMRC} showed that all ordering CSPs (a variant of CSP) do not admit non-trivial approximation algorithm under the Unique Games Conjecture.

Therefore, it is natural to set the expected number of constraints satisfied by the trivial randomized algorithm as a baseline, namely $AVG$, and ask for better algorithms. In the \emph{constraint satisfaction problem above average}, we are given a CSP instance $\calI$ and a parameter $t$, and the goal is to decide whether there is an assignment satisfying $t$ constraints more than the baseline $AVG$.

The CSP above average problem has been extensively studied in the parameterized algorithms designing research community. Gutin et al. \cite{GKS09}
showed that the \maxthreexor~above average (indeed {\sc Max$d$XOR} for arbitrary $d$) is fixed-parameter tractable (FPT), i.e., there is an algorithm that makes the correct decision in time $f(t) \mathrm{poly}(n)$. Later, Alon et al. \cite{AGKSY} showed that every CSP above average admits an algorithm with runtime $\left(2^{O(t^2)} + O(m)\right)$, and therefore is fixed-parameter tractable. Later, Crowston et al. \cite{CFGJKRRT} improved the running time of \maxxor~to $2^{O(t \log t)} \cdot \mathrm{poly}(nm)$. Recently, Makarychev, Makarychev and Zhou \cite{MMZ15} studied a variant, and showed that the ordering CSP above average is fixed-parameter tractable.

\paragraph{Constraint satisfaction problem with a global cardinality constraint.} Given a boolean CSP instance $\calI$, we can impose a global cardinality                                                                                                                                                                                                                                                              constraint $\sum_{i = 1}^{n} x_i = (1-2p)n$ (we assume that $pn$ is an integer). Such a constraint is called the \emph{bisection constraint} if $p = 1/2$. For example, the \maxbisect~problem is the \maxcut~problem with the bisection constraint. Constraint satisfaction problems with global cardinality constraints are natural generalizations of boolean CSPs. Researchers have been studying approximation algorithms for CSPs with global cardinality constraints for decades, where the \maxbisect~problem \cite{FJ97,Ye01,HZ02,FL06,GMRSZ,RT12,ABG13} and the \SSE~problem \cite{RS10,RST10,RST12} are two prominent examples.

Adding a global cardinality constraint could strictly enhance the hardness of the problem. The \SSE~problem can be viewed as the \mincut~problem with the cardinality of the selected subset to be $\rho|V|$ ($\rho \in (0, 1)$).  While \mincut~admits a polynomial-time algorithm to find the optimal solution, we do not know a good approximation algorithm for \SSE. Raghavendra and Steurer\cite{RS10} suggested that the \SSE~problem is where the hardness of the notorious \UG~problem \cite{khot} stems from.


For many boolean CSPs, via a simple reduction described in \cite{GMRSZ,RT12}, if such a CSP does not admit a non-trivial approximation algorithm, neither does the CSP with the bisection constraint. Therefore, it is natural to set the performance of the trivial randomized algorithm as the baseline, and ask for better algorithms for CSPs with a global cardinality constraint. Specifically, given a CSP instance $\calI$ and the cardinality constraint $\sum_{i = 1}^{n} x_i = (1-2p)n$, we define our baseline $AVG$ to be the expected number of constraints satisfied by uniform randomly choosing an assignment complying with the cardinality constraint. The task is to decide whether there is an assignment (again complying with the cardinality constraint), satisfying at least $AVG + t$ constraints. We call this problem \emph{CSP above average with a global cardinality constraint}, and our goal is to design an FPT algorithm for the problem.



Recently, Gutin and Yeo \cite{GY10} showed that it is possible to decide whether there is an assignment satisfying more than $\lceil m/2 + t \rceil$ constraints in time $\left(2^{O(t^2)} + O(m)\right)$ for the \maxbisect~problem with $m$ constraints and $n$ variables. The running time was later improved to $\left(2^{O(t)} + O(m)\right)$ by Mnich and Zenklusen \cite{MZ12}. However, observe that in the \maxbisect~problem, the trivial randomized algorithm satisfies $AVG = \left(\frac{1}{2} + \frac{1}{2(n-1)}\right)m $ constraints in expectation. Therefore, when $m \gg n$, our problem \maxbisect~above average asks more than what was proved in \cite{GY10, MZ12}. For the \maxcut~problem without any global cardinality constraint, Crowston et al. \cite{CJM15} showed that optimizing above the Edwards-Erd{\H{o}}s bound is fixed-parameter tractable, which is comparable to the bound in our work, while our algorithm outputs a solution strictly satisfying the global cardinality constraint $\sum_{i = 1}^{n} x_i = (1-2p)n$.

\subsection{Our results} In this paper, we show FPT algorithms for boolean CSPs above average with a global cardinality constraint problem. Our main theorem is stated as follows.

\begin{version:10pages}
\begin{theorem}[Informal version of Theorem~\ref{CSP_bisection}]\label{inform_FPT}
For any integer constant $d$ and real constant $p_0 \in (0, 1/2]$, given a $d$-ary CSP with $n$ variables and $m = n^{O(1)}$ constraints, a global cardinality constraint $\sum_{i = 1}^{n} x_i = (1-2p)n$ such that $p \in [p_0, 1 - p_0]$, and an integer parameter $t$, there is an algorithm that runs in time $(n^{O(1)} + 2^{O(t^2)})$ and  decides whether there is an assignment complying with the cardinality constraint to satisfy at least $(AVG + t)$ constraints or not. Here $AVG$ is the expected number of constraints satisfied by uniform randomly choosing an assignment complying with the cardinality constraint.
\end{theorem}
\end{version:10pages}
\begin{version:full}
\begin{theorem}[Informal version of Theorem~\ref{CSP_bisection} and Theorem~\ref{CSP_global}]\label{inform_FPT}
For any integer constant $d$ and real constant $p_0 \in (0, 1/2]$, given a $d$-ary CSP with $n$ variables and $m = n^{O(1)}$ constraints, a global cardinality constraint $\sum_{i = 1}^{n} x_i = (1-2p)n$ such that $p \in [p_0, 1 - p_0]$, and an integer parameter $t$, there is an algorithm that runs in time $(n^{O(1)} + 2^{O(t^2)})$ and  decides whether there is an assignment complying with the cardinality constraint to satisfy at least $(AVG + t)$ constraints or not. Here $AVG$ is the expected number of constraints satisfied by uniform randomly choosing an assignment complying with the cardinality constraint.
\end{theorem}
\end{version:full}

One important ingredient in the proof of our main theorem is the $2\to4$ hypercontractivity of low-degree multilinear polynomials in a correlated probability space. Let $D_p$ be the uniform distribution on all assignments to the $n$ variables complying with the cardinality constraint $\sum_{i = 1}^{n} x_i = (1-2p)n$. We show the following inequality.
\begin{version:10pages}
\begin{theorem}\label{inform_hyper}
For any degree $d$ multilinear polynomial $f$ on variables $x_1, x_2, \dots, x_n$, we have
\[
\E_{D_p} [f^4] \leq \mathrm{poly}(d) \cdot C_{p}^d \cdot \E_{D_p} [f^2]^2 ,
\]
where the constant $C_{p} = \mathrm{poly}(\frac{1}{1-p}, \frac{1}{p})$.
\end{theorem}
\end{version:10pages}
\begin{version:full}
\begin{theorem}[Informal version of Corollary~\ref{cor:hypercontractivity_bisection} and Corollary~\ref{hypercontractivity_global_2nd}]\label{inform_hyper}
For any degree $d$ multilinear polynomial $f$ on variables $x_1, x_2, \dots, x_n$, we have
\[
\E_{D_p} [f^4] \leq \mathrm{poly}(d) \cdot C_{p}^d \cdot \E_{D_p} [f^2]^2 ,
\]
where the constant $C_{p} = \mathrm{poly}(\frac{1}{1-p}, \frac{1}{p})$.
\end{theorem}
\end{version:full}

The ordinary $2\to 4$ hypercontractive inequality (see Section~\ref{sec:pre-fourier} for details of the inequality) has wide applications in computer science, e.g., invariance principles \cite{MOO10}, a lower bound on the influence of variables on Boolean cube \cite{KKL}, and an upper bound on the fourth moment of low degree functions \cite{AGKSY,MMZ15} (see \cite{OD} for a complete introduction and more applications with the reference therein). The inequality admits an elegant induction proof, which was first introduced in \cite{MOO05}; and the proof was later extended to different settings (e.g. to the low-level sum-of-squares proof system \cite{BBH12}, and to more general product distributions \cite{MMZ15}). All the previous induction proofs, to the best of our knowledge, rely on the local independence of the variables (i.e. the independence among every constant-sized subset of random variables). In the $2\to4$ hypercontractive inequality we prove, however, every pair of the random variables is correlated.

Because of the lack of pair-wise independence, our induction proof (as well as the proof to the main theorem (Theorem~\ref{inform_FPT})) crucially relies on the analysis of the eigenvalues of several $n^{O(d)} \times n^{O(d)}$ set-symmetric matrices. We will introduce more details about this analysis in the next subsection.

\paragraph{Independent work on the $2\to 4$ hypercontractive inequality.} Independently, Filmus and Mossel\cite{FM} provided a hypercontractive inequality over $D_p$ based on the log-Sobolev inequality due to Lee and Yau \cite{LY98}. They utilized the property that harmonic polynomials constitute an orthogonal basis in $D_p$. In this work, we use parity functions and their Fourier coefficients to analyze the eigenspaces of $\Var_{D_p}$ and prove the hypercontractivity in $D_p$. Parity functions do not constitute an orthogonal basis in $D_p$, e.g., the $n$ variables are not independent under any global cardinality constraint $\sum_{i=1}^n x_i=(1-2p)n$. However, there is another important component in the proof of our main theorem -- we need to prove the variance of a random solution is high if the optimal solution is much above average. Parity functions also play an important role in this component. More specifically, our analysis relies on the fact that the null space of $\Var_{D_p}$ is equivalent to the subspace spanned by the global cardinality constraint and all non-zero eigenvalues of $\Var_{D_p}$ are $\Theta(1)$ in terms of their Fourier coefficients (respect to parity functions).

\subsection{Techniques and proof overview}
Our high level idea is similar to the framework introduced by Gutin et al. \cite{GKSY09}. However, we employ very different techniques to deal with the fixed-parameter tractability above average under a global cardinality constraint. Specifically, we extensively use the analysis of the eigenspaces of the association schemes. For simplicity, we will first use the bisection constraint $\sum_{i = 1}^n x_i=0$ as the global cardinality constraint to illustrate our main ideas. Then we discuss how to generalize it to a global cardinality constraint specified by $p \in (0,1)$.

We first review the basic ideas underlying the work \cite{AGKSY,MMZ15}. Let $f_{\cal I}$ be the degree $d$ multilinear polynomial counting the number of satisfied constraints in a $d$-ary CSP instance $\cal I$ and $\hf_{\cal I}(S)$ be the coefficient of $\prod_{i \in S} x_i$, i.e., the Fourier coefficient of $\chi_S=\prod_{i \in S} x_i$. We will also view the multilinear polynomial $f$ as a vector in the linear space $\mathrm{span}\left\{\chi_S | S \in {[n] \choose \le d}\right\}$.

We consider the variance of $f_{\calI}$ in the uniform distribution of $\{\pm 1\}^n$, namely $\Var(f_{\cal I})$, and discuss the following 2 cases.
\begin{enumerate}
\item $\Var(f_{\cal I}) = O(t^2)$.  In this case, we would like to reduce the problem to a problem whose size depends only on $t$ (i.e. a small kernel of variables). In  \cite{AGKSY}, this was done because the coefficients in $f_{\cal I}$ are multiples of $2^{-d}$ and the variance $\Var(f_{\cal I}) = \sum_{S \neq \emptyset} \hf_{\cal I}(S)^2$ in the uniform distribution. Then there are at most $O(t^2 2^{2d})$ terms with non-zero coefficients in $f_{\cal I}$, so the size of the kernel is at most  $O(d \cdot t^2 2^{2d})$. One can simply enumerate all possible assignments to the kernel.
\item $\Var(f_{\cal I}) = \Omega(t^2)$. We will claim that the optimal value is at least $AVG+t$. This is done via the $2\to4$ hypercontractive inequality (Theorem~\ref{inform_hyper}), which shows that the low-degree polynomial $f_{\cal I}$ is smooth under the  uniform distribution over all valid assignments. Therefore $f_{\cal I}$ is greater than its expectation plus its standard deviation with positive probability, via the fourth moment method \cite{Berger,AGKSY,OD}.
\end{enumerate}

Now we take a detailed look at the first case. For the uniform distribution $D$ on all assignments complying with the bisection constraint, we no longer have $\Var_D(f_{\cal I}) = \sum_{S \neq \emptyset}\hf_{\cal I}(S)^2$. Indeed, $\Var_D(f_{\cal I})$ might be very different from $\sum_{S \neq \emptyset}\hf_{\cal I}(S)^2$. Let us consider a complete graph in the \maxbisect~problem for example -- the value of a bisection in the complete graph with $n$ vertices is always $\frac{n}{2} \cdot \frac{n}{2}$ in $D$, which indicates that $\Var_D(f_{\cal I}) = 0$. However, $\sum_{S \neq \emptyset}\hf_{\cal I}(S)^2$ is $\Omega(n^2)$. Another example for the \maxbisect~problem is a star graph, i.e., a graph with  $(n-1)$ edges connecting an arbitrary vertex $i$ to the rest of $(n-1)$ vertices. The variance for this instance is also $0$. However, one can check that $\sum_{S \neq \emptyset}\hf_{\cal I}(S)^2 = \Omega(n)$.

To solve the problem above, we first observe that in the star graph instance, the degree 2 polynomial $f_{\calI} =\frac{n}{2} - (\sum_i x_i)\frac{x_1}{2}$ (say the center of the star is at Vertex $1$), which always equals $\frac{n}{2}$ in the support of $D$. This is because $\sum_i x_i = 0$ is a requirement for all valid assignments in $D$. In general, we have $(\sum_i x_i)h = 0$ for all polynomials $h$ under the bisection constraint. Therefore, \[\calS = \left\{(\sum_i x_i)h + c : \mbox{$h$ a degree-at-most $(d-1)$ multilinear polynomial, $c$ constant}\right\}\] constitutes a linear subspace where all the functions in the subspace are equivalent to a constant function in the support of $D$.  Let $B$ be a ${[n] \choose \leq d} \times {[n] \choose \leq d}$ matrix so that $f^T B f = \Var_D(f)$ for all degree-at-most-$d$ multilinear polynomials $f$ (this can be done because $\Var_D(f)$ is just a quadratic polynomial on the Fourier coefficients $\left\{\hf(S)|S \in {[n] \choose \le d}\right\}$). From the discussion above, we know that $\calS$ is in the nullspace of $B$.

\paragraph{Projection onto the orthogonal complement of $\calS$, and the spectral analysis of $B$ via association schemes.} Somewhat surprisingly, we show that  $\calS$ indeed \emph{coincides} with the nullspace of $B$, and the eigenvalues of $B$ of the orthogonal complement of $\calS$ are $\Theta(1)$ (when we treat $d$ as a constant). Once we have this, we let $(\sum_i x_i)h_f $ be the projection of $(f_{\cal I} - \hf_{\calI}(\emptyset))$  onto $\calS = \mathrm{span}\left\{1, (\sum_i x_i)\chi_S | S \in {[n] \choose \le d-1}\right\}$, and have that the 2-norm$^2$ of the projection onto the orthogonal complement of $\calS$,  $\|f_{\cal I} - \hf_{\calI}(\emptyset)-(\sum_i x_i)h_f\|_2^2 = \Theta(\Var_D(f_{\cal I}))$.

To analyze the eigenvalues and eigenspaces of $B$, we crucially use the fact that $B$ is set-symmetric, i.e. the value of $B_{S, T}$ (where $S, T \in {n \choose \leq d}$) only depends on $|S|$, $|T|$, and $|S \cap T|$. If we only consider the homogeneous submatrices of $B$ (i.e. the submatrix of all rows $S$ and columns $T$ so that $|S| = |T| = d'$ for some $0 \leq d' \leq d$), their eigenvalues and eigenspaces are well understood and easy to characterize as they correspond to a special association shceme called the Johnson scheme \cite{Godsil}.  Then we follow the approach of \cite{Grigoriev} to provide a self-contained description of the eigenspaces and eigenvalues of $B$, which extensively use the property $\sum_i x_i=0$ in the support of $D$. We will finally show that all eigenvalues of $B$ in the orthogonal complement of $\calS$ are between $.5$ and $d$.

We will also let $A$ be the matrix corresponding to the quadratic form $\E_D[f^2]$ and analyze the eigenvalues and eigenvectors of $A$ in a similar fashion, for the purpose of our proof.

\paragraph{Rounding procedure for the projection vector $h_f$.} We would like to use the argument that since $\|f_{\cal I} - \hf_{\calI}(\emptyset)-(\sum_i x_i)h_f\|_2^2 = \Theta(\Var_D(f_{\cal I}))  = O(t^2)$ is small,  there are not so many non-zero Fourier coefficients for the function $\left(f_{\cal I} - (\sum_i x_i)h_f\right)$ (which is equivalent to $f_{\cal I}$ on all valid assignments). Therefore we could reduce the whole problem to a small kernel. However, since $h_f$ might not have any integrality property, we cannot directly say that the Fourier coefficients of $(f_{\cal I} - \hf_{\calI}(\emptyset)-(\sum_i x_i)h_f)$ are multiples of $2^{-d}$, and follow the approach of \cite{AGKSY}.

To solve this problem, we need an extra rounding process. We ``round'' the projection vector $h_f$ to a new vector $h$ so that the coefficients of $h$ are multiples of $1/\Gamma$ (for some big constant $\Gamma$ which only depends on $d$). At the same time, the rounding process guarantees that $\|f_{\cal I} - \hf_{\calI}(\emptyset)-(\sum_i x_i)h\|_2^2 \leq O(\|f_{\cal I} - \hf_{\calI}(\emptyset)-(\sum_i x_i)h_f\|_2^2)$, so that we can use the argument of \cite{AGKSY} with $(f_{\cal I} -(\sum_i x_i)h)$ (which is also equivalent to $f$ on all valid assignments) instead.

The rounding process proceeds in a iterative manner. At each iteration, we only round the degree-$d'$ homogeneous part of $h_f$ (starting from $d' = d - 1$, and decrease $d'$ after each iteration). We would like to prove the rounding error $\|(\sum_i x_i)h - (\sum_i x_i)h_f\|_2^2$ is bounded by $O(\|f_{\cal I} - \hf_{\calI}(\emptyset)-(\sum_i x_i)h_f\|_2^2)$, so that $\|f_{\cal I} - \hf_{\calI}(\emptyset)-(\sum_i x_i)h\|_2^2 \leq 2(\|f_{\calI}  - \hf_{\calI}(\emptyset)- (\sum_i x_i)h_f\|_2^2 + \|(\sum_i x_i)h - (\sum_i x_i)h_f\|_2^2) \leq O(\|f_{\cal I} - \hf_{\calI}(\emptyset)-(\sum_i x_i)h_f\|_2^2)$. We will first claim that the coefficients of $h_f$ are ``quite'' close to those of $h$, where, however, such closeness does not directly guarantee that the rounding error is well bounded. Then, using this closeness, we will bound the rounding error via a different argument.

\paragraph{The $2\to4$ hypercontractive inequality for low-degree multilinear polynomials under $D$.} For the second case, we need to prove the $2\to4$ hypercontractive inequality for low-degree multilinear polynomials in $D$, i.e. $\E_D [f^4] = O(1) \E_D [f^2]^2$ for any low-degree multilinear polynomial $f$. We will use the special property of the bisection constraint and reduce the task to the ordinary $2\to4$ hypercontractive inequality (on the uniform production distribution).

Specifically, we view the sampling process in $D$ as first uniformly sampling a perfect matching $M$ among the $n$ variables then assigning $+1$ and $-1$ to the two vertices in each pair of matched variables independently. Observe that once we have the perfect matching, the second sampling step is a product distribution over $n/2$ unbiased coins. Let $f_M$ be the function on $n/2$ variables so that each pair of matched variables always take opposite values, we have $\E_D[f^4] = \E_M \E [f_M^4]$, where the second expectation is over the uniform distribution.

Now we can directly apply the ordinary hypercontractive inequality to $\E[f_M^4]$ and have $\E_D[f^4] = O(1) \E_M \E[f_M^2]^2$. Since $\E_D[f^2]^2 = (\E_M \E[f_M^2])^2$, it remains to show that $\E_M \E[f_M^2]^2 = O(1) (\E_M \E[f_M^2])^2$. This final step, can be viewed as proving the ``$1\to2$ hypercontractivity of $\E[f_M^2]$''. We prove this by analyzing the Fourier coefficients of $f$.

\paragraph{Generalization to the general cardinality constraint $\sum_i x_i = (1 - 2p)n$ via random restriction.} We now discuss how to generalize the parameterized algorithm to the general cardinality constraint $\sum_i x_i = (1 - 2p)n$ for any $p \in [p_0, 1 - p_0]$ (where $p_0 > 0$ is a constant). Let $D_p$ be the uniform distribution on all valid assignments complying with the global cardinality constraint. We first focus on Case 1: $\Var_{D_p} (f_{\calI}) = O(t^2)$.

There are two natural approaches to generalize our previous argument to $D_p$. The first idea is to define $\phi_i = \sqrt{p/(1-p)}$ when $x_i = 1$ and $\phi_i = -\sqrt{(1-p)/p}$ when $x_i =- 1$ so that $\E \phi_i = 0$ and $\E \phi_i^2 = 1$, and do the analysis for $g(\phi_1, \dots, \phi_n) = f_{\calI}(x_1, \dots, x_n)$. The second idea is to work with $x_i = \pm 1$, and try to generalize the analysis. However, we adopt neither of the two approaches because of the following reasons.

For the first idea, indeed one can prove that there exists some polynomial $h_g$, so that $\|g - \hat{g}(\emptyset)-(\sum_i \phi_i)h_g\|_2^2 = \Theta(\Var_{D_p}(f_{\cal I}))$. However, since the Fourier coefficients of $g$ have factors such as ${p}^{1/2}, p, p^{3/2}$, etc., it is not clear why the Fourier coefficients of $(g- \hat{g}(\emptyset)-(\sum_i \phi_i)h)$ are multiples of some number even when we round $h_g$ to $h$. For the second idea, the super-constant on the right-hand-side of the constraint $\sum_i x_i = (1 - 2p)n$ imposes a technical difficulty for bounding the rounding error.

Our final approach for the general cardinality constraint is via a reduction to the existing algorithm for the bisection constraint, using random restriction. Let us assume $p = .49$ to illustrate the high-level idea (and it works for any $p$). Given the constraint $\sum_i x_i = .02n$, we randomly choose a set $Q$ of $.02n$ variables and fix these variables to 1. For any valid assignment, we see that the remaining $.98n$ variables satisfy the bisection constraint $\sum_{i \in \overline{Q}} x_i = 0$. Let $f_Q$ be the function on the remaining $.98n$ variables derived by fixing the variables in $Q$ to $1$ for $f_{\calI}$, i.e. let $f_Q(x_{\overline{Q}}) = f_{\calI}(x_{\overline{Q}}, x_{Q} = \vec{1})$, where $x_Q$ means $\{x_i : i \in Q\}$ and $x_{\overline{Q}}$ means $\{x_i : i \in \overline{Q}\}$.

Since sampling $x$ from $D_p$ is equivalent to sampling $Q$ first then sampling $x_{\overline{Q}} \sim D$ and setting $x_{Q} = \vec{1}$, we have $\E_Q \Var_{D} [f_Q] \leq \Var_{D_p} [f_{\calI}] \leq O(t^2)$. Therefore, for most $Q$'s, we have $ \Var_{D} [f_Q] \leq O(t^2)$. Using our parameterized algorithm for the bisection constraint, we know that for most $Q$'s, $f_Q$ depends on merely $O(t^2)$ variables (on all assignments complying with the bisection constraint). Then we manage to show that to make this happen, a low-degree polynomial $f$ itself has to depend on only $O(t^2)$ variables (on all valid assignments in $D_p$).

\paragraph{The $2\to4$ hypercontractive inequality under distribution $D_p$.} To deal with Case 2, we need the $2\to4$ hypercontractive inequality for low-degree multilinear polynomials in $D_p$. We let  $\phi_i = \sqrt{p/(1-p)}$ when $x_i = 1$ and $\phi_i = -\sqrt{(1-p)/p}$ when $x_i =- 1$ so that $\E \phi_i = 0, \E \phi_i^2 = 1$, and $\sum_i \phi_i=0$ in the support of $D_p$ then prove the inequality for multilinear polynomials on $\phi_1, \phi_2, \dots, \phi_n$. We follow the paradigm introduced in \cite{MOO05} and apply induction on the number of variables and the degree of the polynomial. For any multilinear polynomial $f$, we write it in the form of $f = \phi_1 h_0 + h_1$ (say the first variable $f$ depends on is $\phi_1$), expand $\E_{D_p} [f^4] = \E_{D_p} [( \phi_1 h_0 + h_1)^4]$, and control each term using induction hypothesis separately.

Unlike the proof in \cite{MOO05}, we no longer have $\E_{D_p} [\phi_1 h_0 h_1^3] = 0$ or $\E_{D_p} [\phi_1^3 h_0^3 h_1 = 0]$ because of the lack of independence. While the latter one is easy to deal with, the main technical difficulty comes from the first term, namely upper-bounding $\E_{D_p} [\phi_1 h_0 h_1^3]$ in terms of $\|h_0\|_2$ and $\|h_1\|_2$. We reduce this problem to upper-bounding the spectral norm of a carefully designed set-symmetric matrix $L$, which again utilizes the analysis developed in Section~\ref{eigenvalues}.

\subsection{Organization}
\begin{version:full}
This paper is organized as follows. We introduce some preliminaries in Section~\ref{sec_pre}. In Section~\ref{eigenvalues}, we provide a self-contained description of the eigenspaces and eigenvalues of the matrices corresponding to $E_{D_p}[f^2]$ and $\Var_{D_p}(f)$. In Section~\ref{sec_bisection}, we prove the fixed-parameter tractability of CSPs above average with the bisection constraint, and the $2 \to 4$ hypercontractive inequality under the distribution conditioned on the bisection constraint. In Section~\ref{sec_hyper}, we prove the more general $2 \to 4$ hypercontractive inequality under distribution $D_p$. At last, we show the fixed-parameter tractability of any $d$-ary CSP above average with a global cardinality constraint  in Section~\ref{sec_FPT_global}.
\end{version:full}
\begin{version:10pages}
This paper is organized as follows. We introduce some preliminaries in Section~\ref{sec_pre}. In Section~\ref{eigenvalues}, we provide a self-contained description of the eigenspaces and eigenvalues of the matrices corresponding to $\Var_{D_p}(f)$. In Section~\ref{sec_bisection}, we prove the fixed-parameter tractability of CSPs above average with the bisection constraint, and the $2 \to 4$ hypercontractive inequality under the distribution conditioned on the bisection constraint.

Due to space constraints, we omit much content of Section~\ref{sec_pre}, and most content of Section~\ref{eigenvalues} to Section~\ref{sec_bisection}. We defer the more general $2 \to 4$ hypercontractive inequality under distribution $D_p$ and the fixed-parameter tractability of any $d$-ary CSP above average with a global cardinality constraint to the full version of this paper. 
\end{version:10pages}

\section{Preliminaries}\label{sec_pre}
Let $[n]=\{1,2,\cdots,n\}$. For convenience, we always use ${[n] \choose d}$ to denote the set of all subsets of size $d$ in $[n]$ and ${[n] \choose \le d}$ to denote the set of all subsets of size at most $d$ in $[n]$ (including $\emptyset$). For two subsets $S$ and $T$, we use $S \setdiff T$ to denote the symmetric difference of $S$ and $T$.  Let $n!$ denote the product $\prod_{i=1}^n i$ and $n!!=\prod_{i=0}^{[\frac{n-1}{2}]}(n-2i)$. We use $\vec{0}$ ($\vec{1}$ resp.) to denote the all 0 (1 resp.) vector. We also use $1_{E}$ to denote the indicator variable of an event $E$, i.e. $1_E = 1$ when $E$ is true, and $1_E = 0$ otherwise.

In this work, we only consider $f:\{\pm 1\}^n \rightarrow \mathbb{R}$. Let $U$ denote the uniform distribution on $\{\pm 1\}^n$ and $U_p$ denote the biased product distribution on $\{\pm 1\}^n$ such that each bit equals to $-1$ with probability $p$ and equals to 1 with probability $1-p$. For a distribution $V$, we use $supp(V)$ to denote the support of $V$.

For a random variable $X$ with standard deviation $\sigma$, it is known that the fourth moment is necessary and sufficient to guarantee that there exists $x \in supp(X)$ greater than $\E[X]+\Omega(\sigma)$ from \cite{Berger,AGKSY,OD}. We state this result as follows.
\begin{lemma}\label{4th_moment_method}
Let $X$ be a real random variable. Suppose that $\E[X]=0, \E[X^2]=\sigma^2>0$, and $\E[X^4]<b \sigma^4$ for some $b>0$. Then $\Pr[X \ge \sigma/(2\sqrt{b})]>0$.
\end{lemma}

A global cardinality constraint defined by a parameter $0 < p < 1$ is $\sum_{i \in [n]}x_i=(1-2p)n$, which indicates that $(1-p)$ fraction of $x_i$'s are $1$ and the rest $p$ fraction are $-1$. For convenience, we call the global cardinality constraint with $p=1/2$ as the bisection constraint. In this paper, we always use $D$ to denote the uniform distribution on all assignments to the $n$ variables complying with the bisection constraint $\sum_{i = 1}^{n} x_i =0$ and $D_p$ to denote the uniform distribution on all assignments complying with the cardinality constraint $\sum_{i = 1}^{n} x_i = (1-2p)n$.

\subsection{Basics of Fourier Analysis of Boolean functions}\label{sec:pre-fourier}
We state several basic properties of the Fourier transform for Boolean functions those will be useful in this work. We follow the notations in \cite{OD} except for $q=\frac{2p-1}{\sqrt{p(1-p)}}$ instead of $1-p$.

We first introduce the standard Fourier transform in $\{\pm 1\}^n$. We will also use the $p$-biased Fourier transform in several proofs especially for the $2 \to 4$ hypercontractive inequality under $D_p$. More specifically, in Section~\ref{subsec_distributions}, Section~\ref{eigenvalues}, and Section~\ref{sec_hyper}, we will use the Fourier transform with the $p$-biased basis $\{\phi_S\}$. In Section~\ref{sec_bisection} and Section~\ref{sec_FPT_global}, we will use the standard Fourier transform in the basis $\{\chi_S\}$.

For the uniform distribution $U$, we define the inner-product on a pair of functions $f,g : \{\pm 1\}^n \rightarrow \mathbb{R}$ by $\langle f,g \rangle=\E_{x \sim U}[f(x)g(x)]$. Hence $\chi_S(x)=\prod_{i \in S}x_i$ over all subsets $S \subseteq [n]$ constitute an orthonormal basis for the functions from $\{\pm 1\}^n$ to $\mathbb{R}$. For simplicity, we abuse the notation by writing $\chi_S$ instead of $\chi_S(x)$. Hence every Boolean function has a unique multilinear polynomial expression $f=\sum_{S \subseteq [n]}\hf(S) \chi_S$, where $\hf(S)=\langle f, \chi_S \rangle$ is the coefficient of $\chi_S$ in $f$. In particular, $\hf(\emptyset)=\E_{x\sim U}[f(x)]$. An important fact about Fourier coefficients is Parseval's identity, i.e., $\sum_{S}\hf(S)^2=\E_{x \sim U}[f(x)^2]$, which indicates $\Var_{U}(f)=\sum_{S \neq \emptyset}\hf(S)^2$.

Given any Boolean function $f$, we define its degree to be the largest size of $S$ with non-zero Fourier coefficient $\hf(S)$. In this work, we focus on the multilinear polynomials $f$ with degree-at-most $d$. We use the Fourier coefficients of weight $i$ to denote all Fourier coefficients $\{\hf(S)|S \in {[n] \choose i}\}$ of size $i$ character functions. For a degree-at-most $d$ polynomial $f$, we abuse the notation $f$ to denote a vector in the linear space $span\{\chi_S|S \in {[n] \choose \le d}\}$, where each coordinate corresponds to a character function $\chi_S$ of a subset $S$.

We state the standard Bonami Lemma for Bernoulli $\pm 1$ random variables \cite{Bonami,OD}, which is also known as the $2\to4$ hypercontractivity for low-degree multilinear polynomials.
\begin{lemma}
Let $f:\{-1,1\}^n \rightarrow \mathbb{R}$ be a degree-at-most $d$ multilinear polynomial. Let $X_1,\cdots,X_n$ be independent unbiased $\pm 1$-Bernoulli variables. Then $$\E[f(X_1,\cdots,X_n)^4]\le 9^d \cdot \E[f(X_1,\cdots,X_n)^2]^2.$$
\end{lemma}

For the $p$-biased distribution $U_p$, we define the inner product on pairs of function $f,g : \{\pm 1\}^n \rightarrow \mathbb{R}$ by $\langle f,g \rangle=\E_{x \sim U_p}[f(x)g(x)]$. Then we define $\phi_i(x)=\sqrt{\frac{p}{1-p}} 1_{x_i=1}- \sqrt{\frac{1-p}{p}}1_{x_i=-1}$ and $\phi_S(x)=\prod_{i \in S} \phi_i(x)$. We abuse the notation by writing $\phi_S$ instead of $\phi_S(x)$. It is straightforward to verify $\E_{U_p}[\phi_i]=0$ and $\E_{U_p}[\phi^2_i]=1$. Notice that $\phi_S \phi_T \neq \phi_{S \Delta T}$ unlike $\chi_S \chi_S = \chi_{S \setdiff T}$ for all $x$. However, $\langle \phi_S,\phi_T \rangle=0$ for different $S$ and $T$ under $U_p$. Thus we have the biased Fourier expansion $f(x)=\sum_{S \subseteq [n]}\hf(S) \phi_S(x)$, where $\hf(S)=\langle f,\phi_S \rangle$ in $U_p$. We also have $\hf(\emptyset)=\E_{U_p}[f]$ and Parseval's identity $\sum_{S}\hf(S)^2=\E_{ U_p}[f^2]$, which demonstrates $\Var_{U_p}(f)=\sum_{S \neq \emptyset}\hf(S)^2$. We state two facts of $\phi_i$ that will be useful in the later section.
\begin{enumerate}
\item $x_i=2\sqrt{p(1-p)} \cdot \phi_i + 1-2p$. Hence $\sum_i \phi_i(x)=0$ for any $x$ satisfying $\sum_i x_i=(1-2p)n$.
\item $\phi^2_i=q \cdot \phi_i + 1$ for $q=\frac{2p-1}{\sqrt{p(1-p)}}$. Thus we always write $f$ as a multilinear polynomial of $\phi_i$.
\end{enumerate}

Observe that the largest size of $|T|$ with non-zero Fourier coefficient $\hf(T)$ in the basis $\{\phi_S|S \in {[n] \choose \le d}\}$ is equivalent to the degree of $f$ defined in $\{\chi_S|S \in {[n] \choose \le d}\}$. Hence we still define the degree of $f$ to be $\max_{S:\hf(S)\neq 0} |S|$.  We abuse the notation $f$ to denote a vector in the linear space $span\{\phi_S|S \in {[n] \choose \le d}\}$.

For the biased distribution $U_p$, we know $\E_{U_p}[\phi_i^4]=\frac{p^2}{1-p}+\frac{(1-p)^2}{p}\ge 1$. Therefore we state the $2 \to 4$ hypercontractivity in the biased distribution $U_p$ as follows.
\begin{lemma}
Let $f:\{-1,1\}^n \rightarrow \mathbb{R}$ be a degree-at-most $d$ multilinear polynomial of $\phi_1,\cdots,\phi_n$. Then $$\E_{U_p}[f(X_1,\cdots,X_n)^4]\le \left(9 \cdot \frac{p^2}{1-p}+ 9 \cdot \frac{(1-p)^2}{p}\right)^d \cdot \E_{U_p}[f(X_1,\cdots,X_n)^2]^2.$$
\end{lemma}

At last, we notice that the definition of $\phi_S$ is consistent with the definition of $\chi_S$ when $p=1/2$. When the distribution $U_p$ is fixed and clear, we use $\|f\|_2=E_{x \sim U_p}[f(x)^2]^{1/2}$ to denote the $L_2$ norm of a Boolean function $f$. From Parseval's identity, $\|f\|_2$ is also $(\sum_{S} \hf(S)^2)^{1/2}$. From the Cauchy-Schwarz inequality, one useful property is $\|fg\|_2 \le \|f^2\|_2^{1/2} \|g^2\|_2^{1/2}$.

\subsection{Distributions conditioned on global cardinality constraints}\label{subsec_distributions}
We will study the expectation and the variance of a low-degree multilinear polynomial $f$ in $D_p$. Because $\phi_S$ is consistent with $\chi_S$ when $p=1/2$, we fix the basis to be $\phi_S$ of the $p$-biased Fourier transform. Because $p \in [p_0,1-p_0]$, we treat $q=\frac{2p-1}{\sqrt{p(1-p)}}$ as a constant and hide it in the big-Oh notation.

We first discuss the expectation of $f$ under $D_p$. Because $\E_{D_p}[\phi_S]$ is not necessary 0 for any non-empty subset $S$, $\E_{D_p}[f] \neq \hf(\emptyset)$. Let $\delta_S=\E_{D_p}[\phi_S]$. From symmetry, $\delta_S=\delta_{S'}$ for any $S$ and $S'$ with the same size. For convenience, we use $\delta_k=\delta_{S}$ for any $S \in {[n] \choose k}$. From the definition of $\delta$, we have $\E_{D_p}[f]=\sum_{S} \hf(S) \cdot \delta_S$.

For $p=1/2$ and $D$, $\delta_k=0$ for all odd $k$ and $\delta_k=(-1)^{k/2}\frac{(k-1)!!}{(n-1) \cdot (n-3) \cdots (n-k+1)}$ for even $k$. We calculate it this way: pick any $T \in {[n] \choose k}$ and consider $\E_D[(\sum_i x_i)\chi_T]=0$. This indicates
$$k \cdot \delta_{k-1} + (n-k)\delta_{k+1}=0.$$
From $\delta_0=1$ and $\delta_1=0$, we could obtain $\delta_k$ for every $k>1$.

For $p\neq 1/2$ and $D_p$ under the global cardinality constraint $\sum_{i \in n}x_i=(1-2p)n$, we consider $\E_{D_p}[\phi_S]$, because $\sum_{i \in n}x_i=(1-2p)n$ indicates $\sum_i \phi_i=0$. Thus we use $\delta_S=\E_{D_p}[\phi_S]$ and calculate it as follows: pick any $T \in {[n] \choose k}$ and consider $\E_{D_p}[(\sum_i \phi_i)\phi_T]=0$. $\phi_i \phi_T=\phi_{T \cup i}$ for $i \notin T$; and $\phi_i \phi_T = q \cdot \phi_T + \phi_{T \setminus i}$ for $i\in T$ from the fact $\phi_i^2=q \cdot \phi_i + 1$. We have \begin{equation}\label{biased_delta}
k \cdot \delta_{k-1} + k \cdot q \cdot\delta_k + (n-k) \delta_{k+1}=0
\end{equation}
\begin{remark}
For $p=1/2$ and the bisection constraint, $q=0$ and the recurrence relation becomes $k \cdot \delta_{k-1} + (n-k)\delta_{k+1}=0$, which is consistent with the above characterization. Thus we abuse the notation $\delta_k$ when $U_p$ is fixed and clear.
\end{remark}

From $\delta_0=1, \delta_1=0,$ and the relation above, we can determine $\delta_k$ for every $k$. For example, $\delta_2=-\frac{1}{n-1}$ and  $\delta_3 = - \frac{\delta_2}{n-2} \cdot 2 \cdot q=\frac{2 q}{(n-1)(n-2)}$. We bound $\delta_i$ as follows:
\begin{claim}\label{bound_delta}
For any $i \ge 1$, $\delta_{2i-1}=(-1)^i O(n^{-i})$ and $\delta_{2i}=(-1)^i \frac{(2i-1)!!}{n^i}+O(n^{-i-1})$.
\end{claim}
\begin{proof}
We use induction on $i$. Base Case: $\delta_0=1$ and $\delta_1=0$.

Because $\delta_{2i-2}=(-1)^{i-1} \Theta(n^{-i+1})$ and $\delta_{2i-1}=(-1)^i O(n^{-i})$, the major term of $\delta_{2i}$ is determined by $\delta_{2i-2}$. We choose $k=2i-1$ in the equation \eqref{biased_delta} to obtain $$\delta_{2i}=(-1)^i \frac{(2i-1)!!}{(n-1)(n-3)\cdots (n-2i+1)}+O(\frac{1}{n^{i+1}})=(-1)^i \frac{(2i-1)!!}{n^i}+O(\frac{1}{n^{i+1}}).$$ At the same time, from $\delta_{2i}$ and $\delta_{2i-1}$, $$\delta_{2i+1}=(-1)^{i+1}q \cdot \frac{(2i)(2i-1)!!+(2i)(2i-2)(2i-3)!! + \cdots+ (2i)!!}{n^{i+1}}+O(\frac{1}{n^{i+2}})=(-1)^{i+1} O(\frac{1}{n^{i+2}}).$$
\end{proof}

Now we turn to $\E_{D_p}[f^2]$ and $\Var_{D_p}[f]$ for a degree-at-most-$d$ multilinear polynomial $f$. From the definition and the Fourier transform $f=\sum_{S}\hf(S) \phi_S$,
$$\E_{D_p}[f^2]=\sum_{S,T}\hf(S)\hf(T)\delta_{S \setdiff T}, \quad \Var_{D_p}(f)=\E_{D_p}[f^2]-\E_{D_p}[f]^2=\sum_{S,T}\hf(S)\hf(T)(\delta_{S \setdiff T}-\delta_S \delta_T).$$
We associate a ${n \choose \le d} \times {n \choose \le d}$ matrix $A$ with $\E_{D_p}[f^2]$ that $A(S,T)=\delta_{S \setdiff T}$. Hence $\E_{D_p}[f^2] = f^T \cdot A \cdot f$ from the definition when we think $f$ is a vector in the linear space of $span\{\phi_S|S \in {[n] \choose \le d}\}$.

Similarly, we associate  a ${n \choose \le d} \times {n \choose \le d}$ matrix $B$ with $\Var_{D_p}(f)$ that $B(S,T)=\delta_{S \setdiff T}-\delta_S \cdot \delta_T$. Hence $\Var_{D_p}(f) = f^T \cdot B \cdot f$. Notice that an entry $(S,T)$ in $A$ and $B$ only depends on the size of $S,T,$ and $S \cap T$.
\begin{remark}
Because $B(\emptyset,S)=B(S,\emptyset)=0$ for any $S$ and $\Var_{D_p}(f)$ is independent with $\hf(\emptyset)$, we could neglect $\hf(\emptyset)$ in $B$ such that $B$ is a $({[n] \choose d}+\cdots {[n] \choose 1})\times ({[n] \choose d}+\cdots {[n] \choose 1})$ matrix. $\hf(\emptyset)$ is the only difference between the analysis of eigenvalues in $A$ and $B$. Actually, the difference $\delta_S \cdot \delta_T$ between $A(S,T)$ and $B(S,T)$ will not effect the analysis of their eigenvalues except the eigenvalue induced by $\hf(\emptyset)$.
\end{remark}
In Section \ref{eigenvalues}, we study the eigenvalues of $\E_{D_p}[f^2]$ and $\Var_{D_p}(f)$ in the linear space $span\{\phi_S|S \in {[n] \choose \le d}\}$, i.e., the eigenvalues of $A$ and $B$.

\subsection{Eigenspaces in the Johnson Schemes}\label{association_schemes}
We shall use a few characterizations about the eigenspaces of the Johnson scheme to analyze the eigenspaces and eigenvalues of $A$ and $B$ in Section \ref{eigenvalues} (please see \cite{Godsil} for a complete introduction).

We divide $A$ into $(d+1) \times (d+1)$ submatrices where $A_{i,j}$ is the matrix of $A(S,T)$ over all $S \in {[n] \choose i}$ and $T \in {[n] \choose j}$. For each diagonal matrix $A_{i,i}$, observe that $A_{i,i}(S,T)$ only depends on $|S \cap T|$ because of $|S|=|T|=i$, which indicates $A_{i,i}$ is in the association schemes, in particular, Johnson scheme.
\begin{definition}
A matrix $M \in \mathbb{R}^{{[n] \choose r}\times {[n] \choose r}}$ is set-symmetric if for every $S,T \in {[n] \choose r}$, $M(S,T)$ depends only on the size of $|S \cap T|$.

For $n,r\le n/2$, let $\mathcal{J}_r \subseteq \mathbb{R}^{{[n] \choose r}\times {[n] \choose r}}$ be the subspace of all set-symmetric matrices. $\mathcal{J}_r$ is called the Johnson scheme.
\end{definition}
Let $M \in \mathbb{R}^{{[n] \choose r}\times {[n] \choose r}}$ be a matrix in the Johnson scheme $\mathcal{J}_r$. We treat a vector in $\mathbb{R}^{[n] \choose r}$ as a homogeneous degree $r$ polynomial $f=\sum_{T \in {[n] \choose r}}\hf(T)\phi_T$, where each coordinate corresponds to a $r$-subset. Although the eigenvalues of $M$ depend on the entries of $M$, the eigenspaces of $M$ are independent with $M$ as long as $M$ is in the Johnson scheme.
\begin{fact}
There are $r+1$ eigenspaces $V_0,V_1,\cdots,V_r$ in $M$. For $i \in [r]$, the dimension of $V_i$ is ${n \choose i}-{n \choose i-1}$; and the dimension of $V_0$ is 1. We define $V_i$ through $\hf(S)$ over all $S \in {[n] \choose i}$, although $M$ and $f$ only depend on $\{\hf(T)|T\in {[n] \choose r}\}$. $V_i$ is the linear space spanned by $\{\hf(S)\phi_S| S \in {[n] \choose i}\}$ with the following two properties:
\begin{enumerate}
\item For any $T' \in {[n] \choose i-1}$, $\{\hf(S)|S \in {[n] \choose i}\}$ satisfies that $\sum_{j \notin T'}\hf(T' \cup j)=0$ (neglect this property for $V_0$).
\item For any $T \in {[n] \choose r}$, $\hf(T)=\sum_{S \in {T \choose i}} \hf(S)$.
\end{enumerate}
It is straightforward to verify that the dimension of $V_i$ is ${n \choose i}-{n \choose i-1}$ and $V_i$ is an eigenspace in $M$. Notice that the homogeneous degree $i$ polynomial $\sum_{S \in {[n] \choose i}} \hf(S)\phi_S$ is an eigenvector of matrices in $\mathcal{J}_i$.
\end{fact}
To show the orthogonality between $V_i$ and $V_j$, it is enough to prove that
\begin{claim}\label{orthogonality}
For any $j\le r$ and any $ S \in {[n] \choose <j}$, $\sum_{T \in {[n] \choose j}:S \subset T}\hf(T)=0$ for any $f \in V_j$.
\end{claim}
\begin{proof}
We use induction on the size of $S$ to show it is true.

Base Case $|S|=j-1$: from the definition of $f$, $\sum_{T:S \subset T}\hf(T)=0$.

Suppose $\sum_{T: S \subset T} \hf(T)=0$ for any $S \in {[n] \choose k+1}$. We prove it is true for any $S\in{[n] \choose k}$:
$$\sum_{T: S \subset T}\hf(T)=\frac{1}{j-|S|}\sum_{i \notin S} \sum_{T: (S \cup i) \subset T}\hf(T)=0.$$
\end{proof}

\subsection{CSPs with a global cardinality constraint}
In this work, we consider the constraint satisfaction problem on $\{-1,1\}^n$ with a global cardinality constraint. We allow different constraints using different predicates. Because we can add dummy variables in each constraint, we assume the number of variables in each constraint is $d$ for simplicity.
\begin{definition}
An instance $\cal I$ of a constraint satisfaction problem of arity $d$ consists of a set of variables $V=\{x_1,\cdots,x_n\}$ and a set of $m$ constraints $C_1,\cdots,C_m$. Each constraint $C_i$ consists of $d$ variables $x_{i_1},\cdots,x_{i_d}$ and a predicate $P_i \subseteq \{-1,1\}^d$. An assignment on $x_{i_1},\cdots,x_{i_d}$ satisfies $C_i$ if and only if $(x_{i_1},\cdots,x_{i_d})\in P_i$.  The value $\val_{\cal I}(\alpha)$ of an assignment $\alpha$ is the number of constraints in $C_1,\cdots,C_m$ that are satisfied by $\alpha$. The goal of the problem is to find an assignment with maximum possible value.

An instance $\cal I$ of a constraint satisfaction problem with a global cardinality constraint consists of an instance $\cal I$ of a CSP and a global cardinality constraint $\sum_{i \in [n]}x_i=(1-2p)n$ specified by a parameter $p$. The goal of the problem is to find an assignment of maximum possible value complying with the global cardinality constraint $\sum_{i \in [n]}x_i=(1-2p)n$. We denote the value of the optimal assignment by $$OPT=\max_{\alpha:\sum_{i} \alpha_i=(1-2p)n} \val_{\cal I}(\alpha).$$ The average value $AVG$ of $\cal I$ is the expected value of an assignment chosen uniformly at random among all assignments complying the global cardinality constraint $$AVG=\E_{\alpha:\sum_{i} \alpha_i=(1-2p)n}[\val_{\cal I}(\alpha)].$$
\end{definition}
Given an instance $\cal I$ of a constraint satisfaction problem of arity $d$, we associate a degree-at-most $d$ multilinear polynomial $f_{\cal I}$ with $\cal I$ such that $f_{\cal I}(\alpha)=\val_{\cal I}(\alpha)$ for any $\alpha \in \{\pm 1\}^n$ as in \cite{AGKSY}.
$$f_{\cal I}(x)=\sum_{i \in [m]}\sum_{\sigma \in P_i}\frac{\prod_{j \in [d]}(1+\sigma_j \cdot x_{i,j})}{2^d}.$$
\begin{remark}
The degree of $f_{\cal I}$ is at most $d$; and the coefficients of $f_{\cal I}$ in the standard basis $\{\chi_S|S \in {[n] \choose \le d}\}$ are always multiples of $2^{-d}$.
\end{remark}
Thus we focus on the study of degree-$d$ polynomial $f$ with coefficients of multiples of $2^{-d}$ instead of the $m$ constraints $C_1,\cdots,C_m$. From the discussion above, given an instance $\cal I$ and a global cardinality constraint $\sum_{i \in n}x_i=(1-2p)n$, the expectation of $\cal I$ under the global cardinality constraint is different than its expectation in the uniform distribution, even for CSPs of arity 2 in the bisection constraint:$$AVG=\E_{D}[f_{\cal I}]=\sum_{S}\hf(S)\E_{D}[\chi_S]=\hf(\emptyset)+\sum_{S \neq \emptyset}\hf(S)\delta_S.$$

\begin{definition}
In the satisfiability above Average Problem, we are given an instance of a CSP of arity $d$, a global cardinality constraint $\sum_{i \in n}x_i=(1-2p)n$, and a parameter $t$. We need to decide whether $OPT\ge AVG+t$ or not.
\end{definition}
In this work, we show that it is fixed-parameter tractable. Namely, given a parameter $t$ and an instance of a CSP problem of arity $d$ under a global cardinality constraint $\sum_{i \in n}x_i=(1-2p)n$, we design an algorithm that either finds a kernel on $O(t^2)$ variables or certifies that $OPT\ge AVG+t$.

\section{Eigenspaces and Eigenvalues of $E_{D_p}[f^2]$ and $\Var_{D_p}(f)$}\label{eigenvalues}
In this section we analyze the eigenvalues and eigenspaces of $A$ and $B$, following the  approach of Grigoriev \cite{Grigoriev}.


We fix any $p \in (0,1)$ with the global cardinality constraint $\sum_i x_i=(1-2p)n$ and use the $p$-biased Fourier transform in this section, i.e., $\{\phi_S|S \in {[n] \choose \le d}\}$. Because $\chi_S$ is consistent with $\phi_S$ for $p=1/2$, it is enough to study the eigenspaces of $A$ and $B$ in $span\{\phi_S|S \in {[n] \choose \le d}\}$. Since $A$ can be divided into $(d+1) \times (d+1)$ submatrices where we know the eigenspaces of the diagonal submatrices from the Johnson scheme, we study the eigenspaces of $A$ through the global cardinality constraint $\sum_i \phi_i=0$ and the relations between eigenspaces of these diagonal matrices characterized in Section \ref{association_schemes}, which is motivated by Grigoriev \cite{Grigoriev}. We will focus on the analysis of $A$ in most time and discuss about $B$ in the end of this section.

We first show the eigenspace $V'_{null}$ with an eigenvalue 0 in $A$, i.e., the null space of $A$. Because $\sum_i x_i=(1-2p)n$ in the support of $D_p$, $\sum_i \phi_i(x)=0$ for any $x$ in the support of $D_p$. Thus $(\sum_i \phi_i)h=0$ for all polynomial $h$ of degree-at-most $d-1$, which is in the linear subspace $span\{(\sum_i \phi_i)\phi_S|S \in {[n] \choose \le d-1}\}$. This linear space is the eigenspace of $A$ with an eigenvalue 0; and its dimension is ${n \choose \le d-1}={[n] \choose d-1} + {[n] \choose d-2} + \cdots {[n] \choose 0}$. By the same reason, $V'_{null}$ is the eigenspace in $B$ with an eigenvalue 0.

Let $V_d$ be the largest eigenspace in $A_{d,d}$ on ${[n] \choose d}\times {[n] \choose d}$. We demonstrate how to find an eigenspace of $A$ based on $V_d$. From the definition of $V_d$, for any $f_d \in V_d$, $f_d$ satisfies that $\sum_{j \notin T}\hf_d(T \cup j)=0$ for any $T \in {[n] \choose d-1}$ from the property of the Johnson scheme. Thus, from Claim~\ref{orthogonality} and the fact that $A(S,T)$ only depends on $|S\cap T|$ given $S \in {[n] \choose i}$ and $T \in {[n] \choose d}$, we know $A_{i,d}f_d=\vec{0}$ for all $i \le d-1$. We construct an eigenvector $f$ in $A$ from $f_d$ as follows: $\hf(S)=0$ for all $S \in {[n] \choose <d}$ and $\hf(T)=\hf_d(T)$ for all $T \in {[n] \choose d}$, i.e., $f=(\vec{0},f_{d})$. It is straightforward to verify that $A (\vec{0},f_d)=\lambda_d (\vec{0},f_d)$, where the eigenvalue $\lambda_d$ is the eigenvalue of $V_d$ in $A_{d,d}$.

Then we move to $V_{d-1}$ in $A_{d,d}$ and illustrate how to use an eigenvector in $V_{d-1}$ to construct an eigenvector of $A$. For any $f_d \in V_{d-1}$, let $f_{d-1}=\sum_{S \in {[n] \choose d-1}}\hf_{d-1}(S)\phi_S$ be the homogeneous degree $d-1$ polynomial such that $f_d=\sum_{T \in {[n] \choose d}} \left( \sum_{S \in {T \choose d-1}}\hf_{d-1}(S) \right) \phi_T$. From Claim~\ref{orthogonality}, $A_{i,d}f_d=0$ for all $i<d-1$ and $A_{i,d-1}f_{d-1}=0$ for all $i<d-2$. Observe that $f_{d-1}$ is an eigenvector of $A_{d-1,d-1}$, although it is possible that the eigenvalue of $f_{d-1}$ in $A_{d-1,d-1}$ is different than the eigenvalue of $f_d$ in $A_{d,d}$. At the same time, from the symmetry of $A$ and the relationship between $f_d$ and $f_{d-1}$, $A_{d-1,d}f_d=\beta_0 f_{d-1}$ and $A_{d,d-1}f_{d-1}=\beta_1 f_d$ for some constants $\beta_0$ and $\beta_1$ only depending on $\delta$ and $d$. Thus we can find a constant $\alpha_{d-1,d}$ such that $(\vec{0},f_{d-1},\alpha_{d-1,d}f_d)$ becomes an eigenvector of $A$.

More directly, we determine the constant $\alpha_{d-1,d}$ from the orthogonality between $(\vec{0},f_{d-1},\alpha_{d-1,d} \cdot f_d)$ and the null space $span\{(\sum_i \phi_i)\phi_S|S \in {[n] \choose \le d-1}\}$. We pick any $T \in {[n] \choose d-1}$ and rewrite $(\sum_i \phi_i)\phi_T=\sum_{j \in T}\phi_{T \setminus j}+ q \cdot |T| \cdot \phi_T + \sum_{j\notin T}\phi_{T \cup j}$. From the orthogonality,
\begin{multline*}
q |T| \cdot \hf_{d-1}(T) + \sum_{j \notin T} \alpha_{d-1,d} \big(\sum_{T' \in {T \cup j \choose d-1}}\hf_{d-1}(T')\big)=0\\\Rightarrow \bigg( q|T|+(n-|T|)\alpha_{d-1,d} \bigg)\hf_{d-1}(T) + \alpha_{d-1,d} \bigg(\sum_{T'' \in {T \choose d-2}} \sum_{j\notin T}\hf_{d-1}(T'' \cup j)\bigg)=0.
\end{multline*}
From the property of $f_{d-1}$ that $\sum_{j \notin T''}\hf_{d-1}(T'' \cup j)=0$ for all $T'' \in {[n] \choose d-2}$, we simplify it to $$\big(q|T| + (n- 2|T|)\alpha_{d-1,d}\big)\hf_{d-1}(T)=0,$$ which determines $\alpha_{d-1,d}=\frac{-(d-1)q}{n-2d+2}$ directly.

Following this approach, we figure out all eigenspaces of $A$ from the eigenspaces $V_0,V_1,\cdots,V_d$ in $A_{d,d}$. For convenience, we use $V'_k$ for $k \le d$ to denote the $k$th eigenspace in $A$ extended by $V_k$ in $A_{d,d}$. We first choose the coefficients in the combination. Let $\alpha_{k,i}=0$ for all $i<k$, $\alpha_{k,k}=1$, and $\alpha_{k,k+1},\cdots,\alpha_{k,d}$ satisfy the recurrence relation (we will show the choices of $\alpha$ later):
\begin{equation}\label{biased_alpha}
i \cdot \alpha_{k,k+i-1} + (k+i)\cdot q \cdot\alpha_{k,k+i} + (n-2k-i)\alpha_{k,k+i+1}=0.
\end{equation}
Then for every $f \in V_k$, the coefficients of $\hf(T)$ over all $T \in {[n] \choose \le d}$ spanned by $\{\hf(S)|S \in {[n] \choose k}\}$ satisfy the following three properties:
\begin{enumerate}
\item $\forall T \in {[n] \choose k-1},\sum_{j \notin T}\hf(T \cup j)=0$ (neglect this property for $V'_0$);
\item $\forall T \in {[n] \choose >k}, \hf(T) =\alpha_{k,|T|} \cdot \sum_{S \in {T \choose k}}\hf(S)$;
\item $\forall T \in {[n] \choose <k}, \hf(T) =0$.
\end{enumerate}
Now we show the recurrence relation of $\alpha_{k,k+i}$ from the fact that $f$ is orthogonal to the null space of $A$. We consider  $(\sum_i \phi_i)\phi_T$ in the null space for a subset $T$ of size $k+i<d$ and simplify $(\sum_i \phi_i)\phi_T$ to $\sum_{j \in T}\phi_{T \setminus j}+ q \cdot |T| \cdot \phi_T + \sum_{j\notin T}\phi_{T \cup j}$. We have
\begin{multline*}
 \sum_{j \in T} \alpha_{k,k+i-1} \sum_{S \in {T\setminus j \choose k}} \hf(S) + (k+i)\cdot q\cdot \alpha_{k,k+i}\sum_{S \in {T \choose k}}\hf(S) + \sum_{j \notin T} \alpha_{k,k+i+1} \sum_{S \in {T \cup j \choose k}}\hf(S)=0 \Rightarrow \\
 \sum_{S \in {T \choose k}}\big( i \cdot \alpha_{k,k+i-1} + (k+i) q \cdot\alpha_{k,k+i} + (n-k-i)\alpha_{k,k+i+1}\big) \hf(S)+ \sum_{T' \in {T \choose k-1}} \alpha_{k,k+i+1} \sum_{j \notin T} \hf(T' \cup j)=0 .
\end{multline*}
Using the first property $\forall T' \in {[n] \choose k-1},\sum_{j \notin T'}\hf(T' \cup j)=0$ to eliminate all $S'$ not in $T$, We obtain
$$\big(i \cdot \alpha_{k,k+i-1} + (k+i)\cdot q \cdot\alpha_{k,k+i} + (n-2k-i)\alpha_{k,k+i+1} \big)\sum_{S \in {T \choose k}}\hf(S)=0.$$ Because $\sum_{S \in {T \choose k}}\hf(S)$ is not necessary equal to 0 to satisfy the first property (actually $\sum_{S \in {T \choose k}}\hf(S)=0$ for all $T \in {[n] \choose k+i}$ indicates $\hf(S)=0$ for all $S \in {[n] \choose k}$),  the coefficient is 0, which provides the recurrence relation in \eqref{biased_alpha}.

The dimension of $V'_k$ is ${[n] \choose k}-{[n] \choose k-1}$ from the first property (It is straightforward to verify $\sum_{k=0}^d dim(V'_k) + dim(V'_{null})=\sum_{k=0}^d({[n] \choose k}-{[n] \choose k-1}) + {[n] \choose d-1} + {[n] \choose d-2} + \cdots {[n] \choose 0}={[n] \choose \le d}$). The orthogonality between $V'_i$ and $V'_j$ follows from Claim \ref{orthogonality} and the orthogonality of $V_i$ and $V_j$.
\begin{remark}
$V'_1,\cdots,V'_d$ are the non-zero eigenspaces in $B$ except for $V'_0$. For $f\in V'_0$, observe that $\hf(T)$ only depends on the size of $T$ and $\hf(\emptyset)$. Hence for any polynomial $f \in V'_0$, $f$ is a constant function over the support of $D_p$, i.e., $\Var_{D_p}(V'_0)=0$. Therefore $V'_0$ is in the null space of $B$.
\end{remark}

We use induction on $i$ to bound $\alpha_{k,k+i}$. From $\alpha_{k,k}=1$ and the recurrence relation \eqref{biased_alpha}, the first few terms would be $\alpha_{k,k+1}=-\frac{k q}{n-2k}$ and $\alpha_{k,k+2}=-\frac{1+(k+1)q \cdot \alpha_{k,k+1}}{n-2k-1} = -\frac{1}{n-2k-1} + O(n^{-2})$. 
\begin{claim}\label{bound_alpha}
$\alpha_{k,k+2i}=(-1)^i \frac{(2i-1)!!}{n^i}+O(n^{-i-1})$ and $\alpha_{k,k+2i+1}=(-1)^{i+1} O(n^{-i-1})$.
\end{claim}
\begin{proof}
We use induction on $i$ again. Base Case: $\alpha_{k,k}=1$ and $\alpha_{k,k+1}=-\frac{k q}{n-2k}.$

From the induction hypothesis $\alpha_{k,k+2i-2}=(-1)^{i-1} \Theta(n^{-i+1})$ and $\alpha_{k,k+2i-1}=(-1)^i n^{-i}$, the major term of $\alpha_{k,k+2i}$ is determined $\alpha_{k,k+2i-2}$ such that $\alpha_{k,k+2i}=(-1)^i \frac{(2i-1)!!}{n^i}+O(n^{-i-1})$. Similarly, $\alpha_{k,k+2i+1}=(-1)^{i+1} O(n^{-i-1})$.
\end{proof}

Now we bound the eigenvalue of $V'_k$. For convenience, we think $0!=1$ and $(-1)!!=1$.
\begin{theorem}\label{eigenvalue_V_k}
For any $k \in \{0,\cdots,d\}$, the eigenvalue of $V'_k$ in $A$ is $\sum_{even \ i=0}^{d-k} \frac{(i-1)!!(i-1)!!}{i!} \pm O(n^{-1})$. For any $k \in \{1,\cdots,d\}$, the eigenvalue of $V'_k$ in $B$ is the same $\sum_{even \ i=0}^{d-k} \frac{(i-1)!!(i-1)!!}{i!} \pm O(n^{-1})$.
\end{theorem}
\begin{proof}
We fix a polynomial $f \in V'_k$ and $S \in {[n] \choose k}$ to calculate $\sum_{T \in {[n] \choose \le d}}A(S,T)\hf(T)$ for the eigenvalue of $V'_k$ in $A$. From the fact $\hf(T)=\alpha_{k,|T|} \cdot \sum_{S' \in {T \choose k}}\hf(S')$,  we expand $\sum_{T}A(S,T)\hf(T)$ into the summation of $\hf(S')$ over all $S' \in {[n] \choose k}$ with coefficients. From the symmetry of $A$, the coefficients of $\hf(S')$ in the expansion only depends on the size of $S \cap S'$ (the sizes of $S$ and $S'$ are $k$). Hence we use $\tau_i$ to denote the coefficients of $\hf(S')$ given $|S' \setdiff S|=i$. Thus $\sum_{T}A(S,T)\hf(T) = \sum_{S' \in {[n] \choose k}}\tau_{|S' \vartriangle S|}\hf(S')$.

We calculate $\tau_0,\cdots,\tau_{2d}$ as follows. Because $|S'|=|S|=k$, $|S \setdiff S'|$ is always even. For $\tau_0$, we only consider $T$ containing $S$ and use $k+i$ to denote the size of $T$.
\begin{equation}\label{coef_tau_0}
\tau_0=\sum_{i=0}^{d-k} {n-k \choose i} \cdot \alpha_{k,k+i} \cdot \delta_i .
\end{equation}
For $\tau_{2l}$, we fix a subset $S'$ with $S \setdiff S'=2l$ and only consider $T$ containing $S'$. We use $k+i$ to denote the size of $T$ and $t$ to denote the size of the intersection of $T$ and $S\setminus S'$.
\begin{equation}\label{coef_tau_2l}
\tau_{2l}=\sum_{i=0}^{d-k}\alpha_{k,k+i} \sum_{t=0}^{i}{l \choose t}{n-k-2l \choose i-t}\delta_{2l+i-2t}.
\end{equation}
We will prove that $\tau_0=\Theta(1)$ and $\tau_{2l}=O(n^{-l})$ for all $l\ge 1$ then eliminate all $S'\neq S$ in $\sum_{T}A(S,T)\hf(T) = \sum_{S' \in {[n] \choose k}}\tau_{S' \vartriangle S}\hf(S')$ to obtain the eigenvalue of $V'_k$.

From Claim \ref{bound_delta} and Claim \ref{bound_alpha}, we separate the summation of $\tau_0=\sum_{i=0}^{d-k} {n-k \choose i} \cdot \alpha_{k,k+i} \cdot \delta_i$ to $\sum_{even \ i}{n-k \choose i} \cdot \alpha_{k,k+i} \cdot \delta_i + \sum_{odd \ i}{n-k \choose i} \cdot \alpha_{k,k+i} \cdot \delta_i$. We replace $\delta_i$ and $\alpha_{k,k+i}$ by the bound in Claim \ref{bound_delta} and Claim \ref{bound_alpha}:
\begin{multline*}
\sum_{even \ i}  {n-k \choose i} (-1)^{\frac{i}{2}+\frac{i}{2}} \left(\frac{(i-1)!!}{n^{\frac{i}{2}}} \frac{(i-1)!!}{n^{\frac{i}{2}}} + O(n^{-i-1})\right)  \\
+ \sum_{odd \ i} {n-k \choose i} (-1)^{\frac{i+1}{2}+\frac{i+1}{2}} \cdot O(n^{-\frac{i+1}{2}}) \cdot O(n^{-\frac{i+1}{2}}).
\end{multline*}
It shows $\tau_0=\sum_{even \ i=0}^{d-k} \frac{(i-1)!!(i-1)!!}{i!} + O(n^{-1})$. For $\tau_{2l}$, we bound it by $O(n^{-l})$ through similar method, where $O(n^{-l})$ comes from the fact that $\alpha_{k,k+i}=O(n^{-\frac{i}{2}})$, ${n-k-2l \choose i-t}<n^{i-t}$, and $\delta_{2l+i-2t}=O(n^{-\frac{2l+i-2t}{2}})$.

At last, we show the eigenvalue of $V'_k$ is $O(1/n)$ close to $\tau_0$, which is enough to finish the proof. From the fact that for any $T' \in {[n] \choose k-1}$, $\sum_{j \notin T'}\hf(T' \cup j)=0$, we have (recall that$|S|=k$)
\begin{multline*}
 (k-i) \left(\sum_{S_0 \in {S \choose i}} \sum_{S_1 \in {[n]\setminus S \choose k-i}}\hf(S_0 \cup S_1)\right) + (i+1) \left(\sum_{S_0 \in {S \choose i+1}}\sum_{S_1 \in {[n]\setminus S \choose k-i-1}}\hf(S_0 \cup S_1)\right) \\=\sum_{S_0 \in {S \choose i}}\sum_{S_1 \in {[n]\setminus S \choose k-i-1}} \sum_{j \notin S_0 \cup S_1} \hf(S_0 \cup S_1 \cup j)=0.
\end{multline*}
Thus we apply it on $\sum_{i}\tau_{2i}\big(\sum_{S' \in {[n] \choose k}: |S \cap S'|=k-i}\hf(S')\big)$ to remove all $S'\neq S$. Let $\tau'_{2k}=\tau_{2k}$ and $$\tau'_{2k-2i-2}=\tau_{2k-2i-2} - \frac{i+1}{k-i} \cdot \tau'_{2k-2i}.$$ Using the above rule, it is straightforward to verify
$$\sum_{i=j}^k\tau_{2i}\bigg(\sum_{S' \in {[n] \choose k}: |S \cap S'|=k-i}\hf(S')\bigg)=\tau'_{2j}\bigg(\sum_{S' \in {[n] \choose k}: |S \cap S'|=k-j}\hf(S')\bigg)$$ from $j=k$ to $j=0$ by induction. Therefore $\sum_{i=0}^d\tau_{2i}\big(\sum_{S' \in {[n] \choose k}: |S \cap S'|=k-i}\hf(S')\big)=\tau'_{0} \hf(S)$. Because $\tau_{2i}=O(n^{-i})$, we have $\tau'_0=\tau_0 \pm O(1/n)$. (Remark: actually, $\tau'_0=\sum_{i=0}^k \tau_{2i} \cdot (-1)^i {k \choose i}.$)

From all discussion above, the eigenvalue of $V'_k$ in $A$ is $\tau'_0$, which is $\sum_{even \ i=0}^{d-k} \frac{(i-1)!!(i-1)!!}{i!} \pm O(n^{-1})$.

For $V'_k$ in $B$ of $k \ge 1$, observe that the difference $\delta_S \cdot \delta_T$ between $A(S,T)$ and $B(S,T)$ will not change the calculation of $\tau$, because $\sum_{T \in {[n] \choose i}}\delta_{S} \delta_T \hf(T)=\delta_{S}\delta_i \big(\sum_{T \in {[n] \choose i}} \hf(T)\big)=0$ from the fact $f$ is orthogonal to $V'_0$.
\end{proof}
Because $\frac{(i-1)!!(i-1)!!}{i!}\le 1$ for any even integer $i \ge 0$, we have the following two corollaries.
\begin{corollary}\label{biased_eigenvalue_2nd_E}
All non-zero eigenvalues of $\E_{D_p}[f^2]$ in the linear space of $span\{\phi_S|S \in {[n] \choose \le d}\}$ are between $.5$ and $[\frac{d}{2}]+1 \le d$.
\end{corollary}
\begin{corollary}\label{biased_eigenvalue_Var}
All non-zero eigenvalues of $\Var_{D_p}[f]$ in the linear space of $span\{\phi_S|S \in {[n] \choose 1,\cdots,d}\}$ are between $.5$ and $[\frac{d+1}{2}] \le d$.
\end{corollary}
Because $f+(\sum_i \phi_i)h \equiv f$ over $supp(D_p)$ for any $h$ of degree-at-most $d-1$, we define the projection of $f$ onto $V'_{null}$ to compare $\|f\|_2^2$ and $\E_{D_p}[f^2]$.
\begin{definition}
Fix the global cardinality constraint $\sum_i x_i=(1-2p)n$ and the Fourier transform $\phi_S$, let $h_f$ denote the projection of a degree $d$ multilinear polynomial $f$ onto the null space $span\{(\sum_i \phi_i)\phi_S|S \in {[n] \choose \le d-1}\}$ of $E_{D_p}[f^2]$ and $\Var_{D_p}(f)$, i.e., $f-(\sum_i \phi_i)h_f$ is orthogonal to the eigenspace of an eigenvalue 0 in $\E_{D_p}[f^2]$ and $\Var_{D_p}(f)$.
\end{definition}
From the above two corollaries and the definition of $h_f$,  we bound $\E_{D_p}[f^2]$ by $\|f\|_2^2$ as follows. For $\Var_{D_p}(f)$, we exclude $\hf(\emptyset)$ because $\Var_{D_p}(f)$ is independent with $\hf(\emptyset)$. Recall that $\|f\|_2^2=E_{U_p}[f^2]=\sum_S \hf(S)^2.$
\begin{corollary}\label{2nd_E_L2_norm}
For any degree $d$ multilinear polynomial $f$ and a global cardinality constraint $\sum_i x_i=(1-2p)n$, $\E_{D_p}[f^2] \le d \|f\|_2^2$ and $\E_{D_p}[f^2] \ge 0.5 \|f-(\sum_i \phi_i)h_f\|_2^2$.
\end{corollary}
\begin{corollary}\label{variance_L2_norm}
For any degree $d$ multilinear polynomial $f$ and a global cardinality constraint $\sum_i x_i=(1-2p)n$, $\Var_{D_p}(f) \le d \|f-\hf(\emptyset)\|_2^2$ and $\Var_{D_p}(f) \ge 0.5 \|f-\hf(\emptyset)-(\sum_i \phi_i)h_{f-\hf(\emptyset)}\|_2^2$.
\end{corollary} 

\section{Parameterized algorithm for CSPs above average with the bisection constraint}\label{sec_bisection}
We prove that CSPs above average with the bisection constraint are fixed-parameter tractable. Given an instance $\cal I$ from $d$-ary CSPs and the bisection constraint $\sum_i x_i=0$, we use the standard basis $\{\chi_S|S \in {[n] \choose \le d}\}$ of the Fourier transform in $U$ and abbreviate $f_{\cal I}$ to $f$. Recall that $\|f\|_2^2=E_U[f^2]=\sum_S \hf(S)^2$ and $D$ is the uniform distribution on all assignments in $\{\pm 1\}^n$ complying with the bisection constraint.

For $f$ with a small variance in $D$, we use $h_{f-\hf(\emptyset)}$ to denote the projection of $f-\hf(\emptyset)$ onto the null space $span\{(\sum_i x_i)\chi_S|S \in {[n] \choose \le d-1}\}$. We know $\|f-\hf(\emptyset)-(\sum_i x_i)h_{f-\hf(\emptyset)}\|_2^2 \le 2 \Var_D(f)$ from Corollary~\ref{variance_L2_norm}, i.e., the lower bound of the non-zero eigenvalues in $\Var_D(f)$. Then we show how to round $h_{f-\hf(\emptyset)}$ in Section~\ref{rounding_bisection} to a degree $d-1$ polynomial $h$ with integral coefficients such that $\|f- \hf(\emptyset) - (\sum_i x_i)h\|_2^2=O(\|f - \hf(\emptyset) - (\sum_i x_i)h_{f-\hf(\emptyset)}\|_2^2)$, which indicates that $f - \hf(\emptyset) - (\sum_i x_i)h$ has a small kernel under the bisection constraint.

Otherwise, for $f$ with a large variance in $D$, we show the hypercontractivity in $D$ that $\E_{D}[(f-\E_D[f])^4]=O(\E_D[(f-\E_D[f])^2]^2)$ in Section~\ref{bisection_hyper}. From the fourth moment method, we know there exists $\alpha$ in the support of $D$ satisfying $f(\alpha)\ge \E_D[f]+\Omega(\sqrt{\Var_D[f]^2})$. At last, we prove the main theorem in Section~\ref{proof_bisection}.
\begin{theorem}\label{CSP_bisection}
Given an instance $\cal I$ of a CSP problem of arity $d$ and a parameter $t$, there is an algorithm with running time $O(n^{\time_h d})$ that either finds a kernel on at most $C_d t^2$ variables or certifies that $OPT \ge AVG+t$ under the bisection constraint for a constant $C_d=24 d^{2} \cdot 7^d \cdot 9^d \cdot 2^{2d} \cdot \big( d!(d-1)!\cdots 2!\big )^2$.
\end{theorem}

\subsection{Rounding}\label{rounding_bisection}
In this section, we show that for any polynomial $f$ of degree $d$ with integral coefficients, there exists an efficient algorithm to round $h_f$ into an integral-coefficient polynomial $h$ while it keeps $\|f-(\sum_i x_i)h\|_2^2=O(\|f-(\sum_i x_i)h_f\|_2^2)$.
\begin{theorem}\label{rounding_h_dCSP}
For any constants $\gamma$ and $d$, given a degree $d$ multilinear polynomial $f$ with $\|f-(\sum_i x_i)h_f\|_2^2 \le \sqrt{n}$ whose Fourier coefficient $\hf(S)$ is a multiple of $\gamma$ for all $S \in {[n] \choose \le d}$, there exists an efficient algorithm to find a degree-at-most $d-1$ polynomial $h$ such that
\begin{enumerate}
\item The Fourier coefficients of $h$ are multiples of $\frac{\gamma}{d!(d-1)!\cdots 2!}$, which demonstrates that the Fourier coefficients of $f - (\sum_i x_i)h$ are multiples of $\frac{\gamma}{d!(d-1)!\cdots 2!}$.
\item $\|f - (\sum_i x_i)h\|_2^2 \le 7^d \cdot \|f - (\sum_i x_i)h_f\|_2^2$.
\end{enumerate}
\end{theorem}
The high level idea of the algorithm is to round $\hh_f(S)$ to $\hh(S)$ from the coefficients of weight $d-1$ to the coefficient of weight 0. At the same time, we guarantee that for any $k<d$, the rounding on the coefficients of weight $k$ will keep $\|f - (\sum_i x_i)h\|_2^2=O(\|f - (\sum_i x_i)h_f\|_2^2)$ in the same order.

Because $h_f$ contains non-zero coefficients up to weight $d-1$, we first prove that we could round $\{\hh_f(S)|S \in {[n] \choose d-1}\}$ to multiples of $\gamma/d!$. Observe that for $T \in {[n] \choose d}$, the coefficient of $\chi_{T}$ in $f-(\sum_i x_i)h_f$ is $\hf(T) - \sum_{j \in T}\hh_f(T \setminus j)$. Because $\sum_{T \in {[n] \choose d}} (\hf(T) - \sum_{j \in T}\hh_f(T \setminus j))^2=o(n)$, $\hf(T) - \sum_{j \in T}\hh_f(T \setminus j)$ is close to 0 for most $T$ in ${[n] \choose d}$. Hence $\sum_{j \in T}\hh_f(T \setminus j) \mod \gamma$ is close to 0 for most $T$. Our start point is to prove that for any $S \in {[n] \choose d-1}$, $\hh(S)$ is close to a multiple of $\gamma/d!$ from the above discussion.
\begin{lemma}\label{eqn_d}
If $\hf(T)$ is a multiple of $\gamma$ and $\hf(T) - \sum_{S \in {T \choose d-1}}\hh_f(S)=0$ for all $T \in {[n] \choose d}$, then $\hh_f(S)$ is a multiple of $\gamma/d!$ for all $S \in {[n] \choose d-1}$.
\end{lemma}
\begin{proof}
From the two conditions, we know $$\sum_{S \in {T \choose d-1}}\hh_f(S) \equiv 0 \mod \gamma$$ for any $T \in {[n] \choose d}.$ We prove that $$(d-1)! \cdot \hh_f(S_1) + (-1)^d (d-1)! \cdot \hh_f(S_2) \equiv 0 \mod \gamma$$ for any $S_1 \in {[n] \choose d-1}$ and $S_2 \in {[n]\setminus S \choose d-1}$. Thus $$0 \equiv (d-1)! \cdot \sum_{S_2 \in {T \choose d-1}}\hh_f(S_2) \equiv d! \cdot \hh_f(S_1)  \mod \gamma,$$ for any $T$ with $S_1 \cap T =\emptyset$, which indicates $\hh_f(S_1)$ is a multiple of $\gamma/d!$.

Without loss of generality, we assume $S_1=\{1,2,\cdots,d-1\}$ and $S_2=\{k_1,k_2,\cdots,k_{d-1}\}$. For a subset $T \in {S _1\cup S_2 \choose d}$, because $\sum_{S \in {T \choose d-1}}\hh_f(S)=\sum_{j \in T}\hh_f(T \setminus j)$, we use $(T)$ to denote the equation
\begin{equation}
\sum_{j \in T}\hh_f(T \setminus j) \equiv 0 \mod \gamma \tag{T}
\end{equation}
Let $\beta_{d-1,1 }=(d-2)!$ and $\beta_{d-i-1, i+1}=\frac{-i}{d-i-1} \cdot \beta_{d-i,i}$ for any $i \in \{1,\cdots,d-2\}$ (we choose $\beta_{d-1,1}$ to guarantee that all coefficients are integers). Consider the following linear combination of equations over $T \in {S_1 \cup S_2 \choose d}$ with coefficients $\beta_{d-i,i}$:
\begin{equation}\label{comb_T}
\sum_{i=1}^{d-1}\beta_{d-i,i} \sum_{T_1 \in {S_1 \choose d-i}, T_2 \in {S_2 \choose i}} (T_1 \cup T_2) \Rightarrow \sum_{i=1}^{d-1}\beta_{d-i,i} \sum_{T_1 \in {S_1 \choose d-i}, T_2 \in {S_2 \choose i}} \big( \sum_{j \in T_1 \cup T_2}\hh_f(T_1 \cup T_2 \setminus j) \big)\equiv 0 \mod \gamma .
\end{equation}
Observe that for any $i\in \{1,\cdots,d-2\}$, $S \in {S_1 \choose d-i-1}$, and $S' \in {S_2 \choose i},$   the coefficient of $\hh_f(S \cup S')$ is $i \cdot \beta_{d-i,i} + (d-i-1) \cdot \beta_{d-i-1,i+1}=0$ in equation \eqref{comb_T}, where $i$ comes from the number of choices of $T_1$ is $d-1-|S|=i$ and $d-i-1$ comes from the number of choices of $T_2$ is $(d-1)-|S'|$.

Hence equation \eqref{comb_T} indicates that $(d-1)\beta_{d-1,1}\hh_f(S_1) + (d-1)\beta_{1,d-1}\hh_f(S_2) \equiv 0 \mod \gamma$. Setting into $\beta_{d-1,1}=(d-2)!$ and $\beta_{1,d-1}=(-1)^{d-2} (d-2)!$, we obtain $$(d-1)! \cdot \hh_f(S_1) + (-1)^d (d-1)! \cdot \hh_f(S_2) \equiv 0 \mod \gamma.$$
\end{proof}
\begin{corollary}\label{integral_d}
If $\sum_{T \in {[n] \choose d}} (\hf(T) - \sum_{S \in {T \choose d-1}}\hh_f(S))^2 = k = o(n^{0.6})$, then for all $S \in {[n] \choose d-1,}, \hh_f(S)$ is $\frac{0.1}{d} \cdot \gamma/d!$ close to a multiple of $\gamma/d!$.
\end{corollary}
\begin{proof}
From the condition, we know that except for $n^{.8}$ choices of $T \in {[n] \choose d}$, $\sum_{S \in {T \choose d-1}}\hh_f(S)$ is $n^{-.1}$ close to a multiple of $\gamma$ because of $n^{.8} \cdot (n^{-.1})^2>k$. Observe that the above proof depends on the Fourier coefficients in at most $2d+1$ variables of $S_1 \cup T$. Because $n^{0.8}=o(n)$, for any subset $S_1 \in {[n] \choose d-1}$, there is a subset $T \in {[n]\setminus S_1 \choose d}$ such that for any $T' \in {S_1 \cup T \choose d}$, $\sum_{S \in {T' \choose d-1}}\hh_f(S)$ is $n^{-.1}$ close to a multiple of $\gamma$.

Following the proof in Lemma \ref{eqn_d}, we obtain that $\hh_f(S)$ is $\frac{(2d)!(d!)^2}{n^{.1}} < \frac{0.1}{d} \cdot \gamma/d!$ close to a multiple of $\gamma/d!$ for any $S \in {[n] \choose d-1}$.
\end{proof}
We consider a natural method to round $h_f$, which is to round $\hh_f(S)$ to the closet multiple of $\gamma/d!$ for every $S \in {[n] \choose d-1}$.
\begin{claim}\label{rounding_process}
Let $h_{d-1}$ be the rounding polynomial of $h_f$ such that $\hh_{d-1}(S)=\hh_f(S)$ for any $|S|\neq d-1$ and $\hh_{d-1}(S)$ is the closest multiple of $\gamma/d!$ to $\hh_f(S)$ for any $S \in {[n] \choose d-1}$. Let $\eps(S)=\hh_{d-1}(S)-\hh_f(S)$.

If $|\eps(S)|<.1/d \cdot \gamma/d!$ and $\alpha(T)$ is a multiple of $\gamma$ for any $T$, then $$\sum_{T \in {[n] \choose d}} \big(\sum_{S \in {T \choose d-1}} \eps(S)\big)^2 \le \sum_{T \in {[n] \choose d}} \big(\alpha(T) - \sum_{S \in {T \choose d-1}} \hh_f(S)\big)^2.$$
\end{claim}
\begin{proof}
For each $T \in {[n] \choose d}$, Because $\sum_{S \in {T \choose d-1}}|\eps(S)|<0.1 \cdot \gamma/d!$, then $|\sum_{S \in {T \choose d-1}} \eps(S)|<|\alpha(T) - \sum_{S \in {T \choose d-1}} \hh_f(S)|$. Hence we know $\sum_{T \in {[n] \choose d}} \big(\sum_{S \in {T \choose d-1}} \eps(S)\big)^2 \le \sum_{T \in {[n] \choose d}} \big(\alpha(T) - \sum_{S \in {T \choose d-1}} \hh_f(S)\big)^2.$
\end{proof}
From now on, we use $h_{d-1}$ to denote the degree $d-1$ polynomial of $h_f$ after the above rounding process on the Fourier coefficients of weight $d-1$. Now we bound the summation of the square of the Fourier coefficients in $f-(\sum_i x_i)h_{d-1}$, i.e., $\|f-(\sum_i x_i)h_{d-1}\|^2_2$. Observe that rounding $\hh_f(S)$ only affect the terms of $T\in {[n] \choose d}$ containing $S$ and $T' \in {[n] \choose d-2}$ inside $S$, because $(\sum_i x_i)\hh_f(S)\chi_S=\sum_{i \in S}\hh_f(S) \chi_{S \setminus i}+\sum_{i \notin S}\hh_f(S) \chi_{S \cup i}$.
\begin{lemma}\label{rounding_onelevel}
$\|f - (\sum_i x_i)h_{d-1} \|_2^2 \le 7\|f - (\sum_i x_i)h_f\|_2^2$.
\end{lemma}
\begin{proof}
Let $\eps(S)=\hh_{d-1}(S)-\hh_f(S)$. It is sufficient to prove
\begin{equation}\label{eq1}
\sum_{T \in {[n] \choose d}} \left(\hf(T) - \sum_{S \in {T \choose d-1}} \hh_f(S) - \sum_{S \in {T \choose d-1}}\eps(S) \right)^2\le 4  \sum_{T \in {[n] \choose d}} \big(\hf(T) - \sum_{S \in {T \choose d-1}} \hh_f(S)\big)^2,
\end{equation} and
\begin{equation}\label{eq2}
\sum_{T' \in {[n] \choose d-2}}\left( \hf(T') - \sum_{S \in {T' \choose d-3}} \hh_f(S) - \sum_{j \notin T'} \hh_f(T' \cup \{j\}) - \sum_{j \notin T'} \eps(T' \cup \{j\})\right)^2 \le 2  \|f - (\sum_i x_i)h_f\|_2^2 .
\end{equation}
Equation \eqref{eq1} follows the fact that $\sum_{T \in {[n] \choose d}} \big(\sum_{S \in {T \choose d-1}} \eps(S)\big)^2 \le \sum_{T \in {[n] \choose d}} \big(\hf(T) - \sum_{S \in {T \choose d-1}} \hh_f(S)\big)^2$ by Claim \ref{rounding_process}. From the inequality of arithmetic and geometric means, we know the cross terms: $$\sum_{T \in {[n] \choose d}} 2 \cdot \big|\hf(T) - \sum_{S \in {T \choose d-1}} \hh_f(S) \big| \cdot |\sum_{S \in {T \choose d-1}} \eps(S)|  \le 2 \sum_{T \in {[n] \choose d}} \big(\hf(T) - \sum_{S \in {T \choose d-1}} \hh_f(S)\big)^2 .$$

For \eqref{eq2}, observe that
\begin{multline*}
\sum_{T' \in {[n] \choose d-2}}\big( \sum_{j \notin T'} \eps(T \cup \{j\}) \big)^2 = (d-1) \sum_{S \in {[n] \choose d-1}} \eps(S)^2 + \sum_{S,S':|S \cap S'|=d-2}2\eps(S)\eps(S') \\ \le \sum_{T \in {[n] \choose d}} \big(\sum_{S \in {T \choose d-1}} \eps(S)\big)^2 \le \sum_{T \in {[n] \choose d}} \big(\hf(T) - \sum_{S \in {T \choose d-1}} \hh_f(S)\big)^2.
\end{multline*}
Hence we have
\begin{multline*}
\sum_{T' \in {[n] \choose d-2}}\big( \hf(T') - \sum_{S \in {T' \choose d-3}} \hh_f(S) - \sum_{j \notin T'} \hh_f(T' \cup \{j\})\big)+ \sum_{T' \in {[n] \choose d-2}}\big( \sum_{j \notin T'} \eps(T \cup \{j\}) \big)^2 \\
\le \sum_{T' \in {[n] \choose d-2}}\big( \hf(T') - \sum_{S \in {T' \choose d-3}} \hh_f(S) - \sum_{j \notin T'} \hh_f(T' \cup \{j\})\big)+ \sum_{T \in {[n] \choose d}} \big(\hf(T) - \sum_{S \in {T \choose d-1}} \hh_f(S)\big)^2 \\
\le \|f - (\sum_i x_i)h_f\|_2^2.
\end{multline*}  We use the inequality of arithmetic and geometric means again to obtain inequality \eqref{eq2}.
\end{proof}
\begin{proofof}{Theorem \ref{rounding_h_dCSP}}
We apply Claim \ref{rounding_process} and Lemma \ref{rounding_onelevel} for $d$ times on the Fourier coefficients of $h_f$ from $\{\hh_f(S)|S \in {[n] \choose d-1 }\},\{\hh_f(S)|S \in {[n] \choose d-2 }\},\cdots$ to $\{\hh_f(S)|S \in {[n] \choose 0 }\}$ by choosing $\gamma$ properly. More specific, let $h_i$ be the polynomial after rounding the coefficients on ${[n] \choose \ge i}$ and $h_d=h_f$. Every time, we use Claim \ref{rounding_process} to round coefficients of $\{\hh_i(S)|S \in {[n] \choose i }\}$ from $h_{i+1}$ for $i=d-1, \cdots, 0$. We use different parameters of $\gamma$ in different rounds: $\gamma$ in the rounding of $h_{d-1}$, $\gamma/d!$ in the rounding of $h_{d-2}$, $\frac{\gamma}{d! \cdot (d-1)!}$ in the rounding of $h_{d-3}$ and so on. After $d$ rounds, all coefficients in $h_0$ are multiples of $\frac{\gamma}{d! (d-1)! (d-2)! \cdots 2!}$.

Because $\|f - (\sum_i x_i)h_i\|_2^2\le 7 \|f - (\sum_i x_i)h_{i+1}\|_2^2$ from Lemma \ref{rounding_onelevel}. Eventually,  $\|f - (\sum_i x_i)h_0\|_2^2\le 7^d \cdot \|f - (\sum_i x_i)h_f\|_2^2.$
\end{proofof}

\subsection{$2\to 4$ hypercontractive inequality under distribution $D$}\label{bisection_hyper}
We prove the $2\to4$ hypercontractivity for a degree $d$ polynomial $g$ in this section.
\begin{theorem}\label{hypercontractivity_bisection}
For any degree-at-most $d$ multilinear polynomial $g$, $\E_D[g^4] \le 3d \cdot 9^{2d} \cdot \|g\|_2^4.$
\end{theorem}
Recall that $\|g\|_2=E_U[g^2]^{1/2}=(\sum_S \hg(S)^2)^{1/2}$ and $g - (\sum_i x_i)h_g \equiv g$ in the support of $D$. Because $\|g-(\sum_i x_i)h_g\|_2^2 \le 2\E_{x \sim D}[g^2]$ from the lower bound of non-zero eigenvalues in $E_D[g^2]$ in Corollary \ref{biased_eigenvalue_2nd_E},  without loss of generality, we assume $g$ is orthogonal to the null space $span\{(\sum_i x_i)\chi_S|S \in {[n] \choose \le d-1}\}$.
\begin{corollary}\label{cor:hypercontractivity_bisection}
For any degree-at-most $d$ multilinear polynomial $g$, $\E_D[g^4] \le 12 d \cdot 9^{2d} \cdot \E_D[g^2]^2.$
\end{corollary}

Before proving the above Theorem, we observe that uniform sampling a bisection $(S,\bar{S})$ is as same as first choosing a random perfect matching $M$ and independently assigning each pair of $M$ to the two subsets. For convenience, we use $P(M)$ to denote the product distribution on $M$ and $\E_M$ to denote the expectation over a uniform random sampling of perfect matching $M$. Let $M(i)$ denote the vertex matched with $i$ in $M$ and $M(S)=\{M(i)|i \in S\}$. From the $2\to4$ hypercontractive inequality on product distribution $P(M)$, we have the following claim:
\begin{claim}\label{claim:hypercon-bisection-1}
$\E_M[ \E_{P(M)}[g^4]]\le 9^d \E_M [\E_{P(M)}[g^2]^2]$.
\end{claim}

Now we prove the main technical lemma of the $2 \to 4$ hypercontractivity under the bisection constraint to finish the proof.
\begin{lemma}\label{lemma:hypercon-bisection-2}
$\E_M [\E_{P(M)}[g^2]^2] \le 3d \cdot 9^d \cdot \|g\|_2^4 $.
\end{lemma}

Theorem \ref{hypercontractivity_bisection} follows from Claim~\ref{claim:hypercon-bisection-1} and Lemma~\ref{lemma:hypercon-bisection-2}.

Now we proceed to the proof of Lemma~\ref{lemma:hypercon-bisection-2}.

\begin{proofof}{Lemma~\ref{lemma:hypercon-bisection-2}}
Using $g(x)=\sum_{S \in {[n] \choose \le d}}\hg(S)\chi_S$, we rewrite $\E_M[\E_{P(M)}[g^2]^2]$ as \begin{multline*}
\E_M\bigg[\E_{P(M)}\big[(\sum_{S \in {[n] \choose \le d}}\hg(S)\chi_S)^2\big]^2\bigg]
\\=\E_M\bigg[\E_{P(M)}\big[\sum_{S \in {[n] \choose \le d}}\hg(S)^2 + \sum_{S\in {[n] \choose \le d},S' \in {[n] \choose \le d}, S'\neq S}\hg(S)\hg(S')\chi_{S \vartriangle S'}\big]^2\bigg].
\end{multline*}
Notice that $\E_{P(M)}[\chi_{S \setdiff S'}] = (-1)^{|S \setdiff S'|/2}$ if and only if $M(S \setdiff S')=S \setdiff S'$; otherwise it is 0. We expand it to
\begin{multline*}
 \E_M\bigg[\Big(\|g\|_2^2 + \sum_{S\in {[n] \choose \le d},S' \in {[n] \choose \le d}, S'\neq S}\hg(S)\hg(S') \cdot 1_{S \vartriangle S'=M(S \vartriangle S')} \cdot (-1)^{|S \vartriangle S'|/2}\Big)^2\bigg]\\
=\|g\|_2^4 + 2 \|g\|_2^2 \cdot \E_M\bigg[\sum_{S\in {[n] \choose \le d},S' \in {[n] \choose \le d}, S'\neq S}\hg(S)\hg(S') \cdot 1_{S \vartriangle S'=M(S \vartriangle S')}\cdot (-1)^{|S \vartriangle S'|/2}\bigg] \\+ \E_M\left[\Big(\sum_{S\in {[n] \choose \le d},S' \in {[n] \choose \le d}, S'\neq S}\hg(S)\hg(S') \cdot 1_{S \vartriangle S'=M(S \vartriangle S')}\cdot (-1)^{|S \vartriangle S'|/2}\Big)^2\right] .
\end{multline*}

We first bound the expectation of $\sum_{S\in {[n] \choose \le d},S' \in {[n] \choose \le d}, S'\neq S}\hg(S)\hg(S') \cdot 1_{S \vartriangle S'=M(S \vartriangle S')}\cdot (-1)^{|S \vartriangle S'|/2}$ in the uniform distribution over all perfect matchings, then bound the expectation of its square. Observe that for a subset $U \subseteq [n]$ with even size, $\E_M[1_{U=M(U)}]= \frac{(|U|-1)(|U|-3)\cdots 1}{(m-1)(m-3)\cdots (m-|U|+1)}$ such that $\E_M[1_{U=M(U)}\cdot (-1)^{|U|/2}]=\delta_{|U|}$, i.e., the expectation $\E_D[\chi_U]$ of $\chi_U$ in $D$. Hence
\[\E_M\big[\sum_{S\in {[n] \choose \le d},S' \in {[n] \choose \le d}, S'\neq S}\hg(S)\hg(S')1_{S \vartriangle S'=M(S \vartriangle S')}\cdot (-1)^{|S \vartriangle S'|/2}\big] \le \sum_{S,S'} \hg(S)\hg(S') \cdot \delta_{S \vartriangle S'}=\E_D[g^2].\]
From Corollary \ref{biased_eigenvalue_2nd_E}, the largest non-zero eigenvalue of the matrix constituted by $\delta_{S \setdiff S'}$ is at most $d$. Thus the expectation is upper bounded by $d \cdot \|g\|_2^2$.

We define $g'$ to be a degree $2d$ polynomial $\sum_{T \in {[n] \choose \le 2d}}\hg'(T)\chi_T$ with coefficients $$\hg'(T)=\sum_{S \in {[n] \choose \le d},S' \in {[n] \choose \le d}: S \setdiff S'=T}\hg(S)\hg(S')$$ for all $T \in {[n] \choose \le 2d}$. Hence we rewrite
\begin{align*}
& \E_M\left[\bigg(\sum_{S\in {[n] \choose \le d},S' \in {[n] \choose \le d}, S'\neq S}\hg(S)\hg(S')\cdot 1_{S \vartriangle S'=M(S \vartriangle S')}\cdot (-1)^{|S \vartriangle S'|/2}\bigg)^2\right]\\
=& \E_M\bigg[\big(\sum_{T\in {[n] \choose \le 2d}}\hg'(T)\cdot 1_{T=M(T)} (-1)^{|T|/2}\big)^2\bigg]\\
=& \E_M\bigg[\sum_{T,T'}\hg'(T)\hg'(T')\cdot 1_{T=M(T)}1_{T'=M(T')}(-1)^{|T|/2+|T'|/2}\bigg].
\end{align*}
Intuitively, because $|T|\le 2d$ and $|T'| \le 2d$, most of pairs $T$ and $T'$ are disjoint such that $\E_M[1_{T=M(T)}1_{T'=M(T')}]=\E_M[1_{T=M(T)}] \cdot \E_M[1_{T'=M(T')}]$. The summation is approximately $\E_M[\sum_T \hg'(T)1_{T=M(T)}(-1)^{|T|/2}]^2$, which is bounded by $d^2 \|g\|_2^4$ from the discussion above. However, we still need to bound the contribution from the correlated paris of $T$ and $T'$.

Notice that $\|g'\|_2^2=\E_U[g^4]$, which can be upper bounded by $\le 9^d \|g\|_2^4$ from the standard $2\to 4$ hypercontractivity.

Instead of bounding it by $\|g\|_2^4$ directly, we will bound it by $2d \cdot \|g'\|_2^2 \le 2d \cdot 9^d \|g\|_2^4$ through the analysis on its eigenvalues and eigenspaces to this end.  For convenience, we rewrite it to $$\E_M\bigg[\sum_{T,T'}\hg'(T)\hg'(T)1_{T=M(T)}1_{T'=M(T')}(-1)^{|T|/2+|T'|/2}\bigg]=\sum_{T,T'} \hg'(T) \hg'(T') \Delta(T,T'),$$ where $\Delta(T,T')=0$ if $|T|$ or $|T'|$ is odd, otherwise
\begin{multline*}
\Delta(T,T')=\E_{M}\bigg[1_{T\cap T'=M(T \cap T')}1_{T=M(T)}1_{T'=M(T')}(-1)^{|T|/2+|T'|/2}\bigg]\\=\frac{|T \cap T'-1|!!\cdot |T \setminus T'-1|!!\cdot |T' \setminus T-1|!!}{(n-1)(n-3)\cdots (n-|T \cup T'|+1)}(-1)^{|T \setdiff T'|/2}.
\end{multline*}
Let $A'$ be the ${n \choose 2d} \times {n \choose 2d}$ matrix whose entry $(T,T')$ is $\Delta(T,T')$. We prove that the eigenspace of $A'$ with eigenvalue 0 is still $span\{(\sum_i x_i)\chi_T|T \in {[n] \choose \le 2d-1}\}$. Because $\Delta_{T,T'}\neq 0$ if and only if $|T|,|T'|,$and $|T \cap T'|$ are even, it is sufficient to show $\sum_{i}A'(S,T \setdiff i)=0$ for all odd sized $T$ and even sized $S$.
\begin{enumerate}
\item $|S \cap T|$ is odd: $\Delta(S, T \vartriangle i) \neq 0$ if and only if $i \in S$.  We separate the calculation into
$i \in S\cap T$ or not: $$\sum_{i}A'(S,T \vartriangle i)\\=\sum_{i \in S \cap T}\Delta(S,T \setminus i)+\sum_{i \in S \setminus T}\Delta(S,T \cup i).$$ Plugging in the definition of $\Delta$, we obtain
\begin{multline*}
\frac{|S\cap T| \cdot |S \cap T-2|!!\cdot |S \setminus T|!!\cdot |T \setminus S|!!}{(n-1)(n-3)\cdots (n-|S \cup T|+1)}(-1)^{|S|/2+|T-1|/2}\\+\frac{|S \setminus T| \cdot |S \cap T|!!\cdot |S \setminus T - 2|!!\cdot |T \setminus S|!!}{(n-1)(n-3)\cdots (n-|S \cup T|+1)}(-1)^{|S|/2+|T+1|/2}=0.
\end{multline*}
\item $|S \cap T|$ is even: $\Delta(S, T \vartriangle i) \neq 0$ if and only if $i \notin S$. We separate the calculation into
$i \in T$ or not:  $$\sum_{i}A'(S,T \vartriangle i)=\sum_{i \in T\setminus S}\Delta(S,T \setminus i)+\sum_{i \notin S \cup T}\Delta(S,T \cup i).$$
Plugging in the definition of $\Delta$, we obtain
\begin{multline*}
\frac{|T \setminus S|\cdot |S \cap T-1|!!\cdot |S \setminus T-1|!!\cdot |T \setminus S - 2|!!}{(n-1)(n-3)\cdots (n-|S \cup T|+1)}(-1)^{|S|/2+|T-1|/2}\\+\frac{(n - |S \cup T|) \cdot |S \cap T - 1|!!\cdot |S \setminus T - 1|!!\cdot |T \setminus S|!!}{(n-1)(n-3)\cdots (n-|S \cup T|)}(-1)^{|S|/2+|T+1|/2}=0 .
\end{multline*}
\end{enumerate}

From the same analysis in Section \ref{eigenvalues}, the eigenspaces of $A'$ are as same as the eigenspaces of $A$ with degree $2d$ except the eigenvalues, whose differences are the differences between $\Delta_{S \vartriangle T}$ and $\delta_{S \vartriangle T}$. We can compute the eigenvalues of $A'$ by the same calculation of eigenvalues in $A$. However, we bound the eigenvalues of $A'$ by $0 \preceq A' \preceq A$ as follows.

Observe that for any $S$ and $T$, $A'(S,T)$ and $A(S,T)$ always has the same sign. At the same time, $|A'(S,T)| = O(\frac{|A(S,T)|}{n^{|S \cap T|}})$. For a eigenspace $V'_k$ in $A$, we focus on $\tau_0$ because the eigenvalue is $O(1/n)$-close to $\tau_0$ from the proof of Theorem \ref{eigenvalue_V_k}. We replace $\delta_i$ by any $\Delta(S,T)$ of $|S|=k,|T|=k+i$ and $|S \cap T|=i$ in $\tau_0=\sum_{i=0}^{d-k} {n-k \choose i} \cdot \alpha_{k,k+i} \cdot \delta_i$ to obtain $\tau'_0$ for $A'$. Thus $\alpha_{k,k+i} \cdot \Delta(S,T) = \Theta(\frac{\alpha_{k,k+i}\delta_i}{n^i})$ indicates $\tau'_0=\Theta(1)<\tau_0$ from the contribution of $i=0$. Repeat this calculation to $\tau_{2l}$, we can show $\tau'_{2l}=O(\tau_{2l})$ for all $l$. Hence we know the eigenvalue of $A'$ in $V'_k$ is upper bounded by the eigenvalue of $A$ from the cancellation rule of $\tau$ in the proof of Theorem~\ref{eigenvalue_V_k}. On the other hand, $A' \succeq 0$ from the definition that it is the expectation of a square term in $M$.

From Corollary \ref{biased_eigenvalue_2nd_E} and all discussion above, we bound the largest eigenvalue of $A'$ by $2d$. Therefore
\begin{multline*}
\E_M\bigg[\sum_{T,T'}\hg'(T)\hg'(T)1_{T=M(T)}1_{T'=M(T')}(-1)^{|T|/2+|T'|/2}\bigg]\\
=\sum_{T,T'} \hg'(T) \hg'(T') \Delta(T,T')\le 2d \|g'\|_2^2 \le 2d \cdot 9^d \cdot \|g\|_2^4.
\end{multline*}

Over all discussion above, $\E_M\big[\E_{P(M)}[g^2]^2\big] \le \|g\|_2^4 + 2d \|g\|_2^4 + 2d \cdot 9^d \cdot \|g\|_2^4 \le 3d \cdot 9^d \cdot \|g\|_2^4$.
\end{proofof}

\subsection{Proof of Theorem \ref{CSP_bisection}}\label{proof_bisection}
In this section, we prove Theorem \ref{CSP_bisection}. Let $f=f_{\cal I}$ be the degree $d$ multilinear polynomial associated with the instance $\cal I$ and $g=f-\E_D[f]$ for convenience. We discuss $\Var_D[f]$ in two cases.

If $\Var_D[f]=\E_D[g^2] \ge 12d \cdot 9^{2d} \cdot t^2$, we have $\E_D[g^4] \le 12d \cdot 9^{2d} \cdot \E_D[g^2]^2$ from the $2\to4$ hypercontractivity of Theorem \ref{hypercontractivity_bisection}. By Lemma \ref{4th_moment_method}, we know $\Pr_D[g \ge \frac{\sqrt{\E_D[g^2]}}{2 \sqrt{12d \cdot 9^{2d}}}]>0.$ Thus $\Pr_D[g \ge t]>0$, which demonstrates that $\Pr_D[f \ge \E_D[f]+t]>0$.

Otherwise we know $\Var_D[f]< 12d \cdot 9^{2d} \cdot t^2$. We consider $f-\hf(\emptyset)$ now. Let $h_{f-\hf(\emptyset)}$ be the projection of $f-\hf(\emptyset)$ onto the linear space such that $\|f-\hf(\emptyset)-(\sum_i x_i)h_{f-\hf(\emptyset)}\|_2^2 \le 2 \Var(f)$ from Corollary \ref{variance_L2_norm} and $\gamma=2^{-d}$. From Theorem \ref{rounding_h_dCSP}, we could round $h_{f-\hf(\emptyset)}$ to $h$ for $f-\hf(\emptyset)$ in time $n^{O(d)}$ such that
\begin{enumerate}
\item the coefficients of $f - \hf(\emptyset)- (\sum_i x_i)h$ are multiples of $\frac{\gamma}{d!(d-1)!\cdots 2!}$;
\item $\|f-\hf(\emptyset) - (\sum_i x_i)h\|_2^2\le 7^d \|f-\hf(\emptyset) - (\sum_i x_i)h_{f-\hf(\emptyset)}\|_2^2 \le 7^d \cdot 2 \cdot 12d \cdot 9^d \cdot t^2$.
\end{enumerate}
We first observe that $f(\alpha)=f(\alpha) - (\sum_i \alpha_i)h(\alpha)$ for any $\alpha$ in the support of $D$. Then we argue $f-\hf(\emptyset) - (\sum_i x_i)h$ has a small kernel, which indicates that $f$ has a small kernel. From the above two properties, we know there are at most $\|f-\hf(\emptyset) -(\sum_i x_i)h\|_2^2/(\frac{\gamma}{d!(d-1)!\cdots 2!})^2$ non-zero coefficients in $f-\hf(\emptyset) - (\sum_i x_i)h$. Because each of the nonzero coefficients contains at most $d$ variables, the instance $\cal I$ has a kernel of at most $24 d^2 \cdot 7^d \cdot 9^d \cdot 2^{2d} \cdot \big( d!(d-1)!\cdots 2!\big )^2 t^2$ variables.

The running time of this algorithm is the running time to find $h_f$ and the rounding time $O(n^d)$. Therefore this algorithm runs in time $O(n^{\time_h d})$.

\section{$2\to4$ hypercontractive inequality under distribution  $D_p$}\label{sec_hyper}
In this section, we prove the $2\to4$ hypercontractivity of low-degree multilinear polynomials in the distribution $D_p$ conditioned on the global cardinality constraint $\sum_i x_i=(1-2p)n$. 

We assume $p$ is in $(0,1)$ such that $p \cdot n$ is a integer. Then we fix the Fourier transform to be the $p$-biased Fourier transform in this section, whose basis is  $\{\phi_S|S \in {[n] \choose \le d}\}$. Hence we use $\phi_1,\cdots,\phi_n$ instead of $x_1,\cdots,x_n$ and say that a function only depends on a subset of characters $\{\phi_i|i \in S\}$ if this function only takes input from variables $\{x_i | i \in S\}$. For a degree $d$ multilinear polynomial $f=\sum_{S \in {[n] \choose \le d}} \hf(S)\phi_S$, we use $\|f\|_2=E_{U_p}[f^2]^{1/2} =(\sum_S \hf(S)^2)^{1/2}$ in this section. 

We rewrite the global cardinality constraint as $\sum_i \phi_i=0$. For convenience, we use $n_{+}=(1-p)n$ to denote the number of $\sqrt{\frac{p}{1-p}}$'s in the global cardinality constraint of $\phi_i$ and $n_{-}=pn$ to denote the number of $-\sqrt{\frac{1-p}{p}}$'s. If $\frac{n_{+}}{n_{-}}=\frac{1-p}{p}=\frac{p_1}{p_2}$ for some integers $p_1$ and $p_2$, we could follow the approach in Section \ref{bisection_hyper} that first partition $[n_{+}+n_{-}]$ into tuples of size $p_1+p_2$ then consider the production distribution over tuples. However, this approach will introduce a dependence on $p_1+p_2$ to the bound, which may be superconstant. Instead of partitioning, we use induction on the number of characters and degree to prove the $2\to4$ hypercontractivity of low-degree multilinear polynomials in $D_p$. 

\begin{theorem}\label{hypercontractivity_global_L2}
For any degree-at-most $d$ multilinear polynomial $f$ on $\phi_1,\cdots,\phi_n$,  $$\E_{D_p}[f(\phi_1, \dots, \phi_n)^4]\le 3 \cdot d^{3/2} \cdot \left(256 \cdot \big((\frac{1-p}{p})^2+(\frac{p}{1-p})^2\big)^2 \right)^d \cdot \|f\|_2^4.$$
\end{theorem}
Recall that $h_f$ is the projection of $f$ onto the null space $span\{(\sum_i \phi_i)\phi_S|S \in {[n] \choose \le d-1}\}$. We know $f-(\sum_i \phi_i)h_f \equiv f$ in $supp({D_p})$, which indicates $\E_{D_p}[f^k]=\E_{D_p}[(f-(\sum_i \phi_i)h_f)^k]$ for any integer $k$. Without loss of generality, we assume $f$ is orthogonal to $span\{(\sum_i \phi_i)\phi_S|S \in {[n] \choose \le d-1}\}$. From the lower bound of eigenvalues in $\E_{D_p}[f^2]$ by Corollary \ref{biased_eigenvalue_2nd_E}, $0.5 \|f\|_2^2 \le \E_{D_p}[f^2]$. We have a direct corollary as follows.
\begin{corollary}\label{hypercontractivity_global_2nd}
For any degree-at-most $d$ multilinear polynomial $f$ on $\phi_1,\cdots,\phi_n$, $$\E_{D_p}[f(\phi_1, \dots, \phi_n)^4] \le 12 \cdot d^{3/2}  \cdot \left(256 \cdot \big((\frac{1-p}{p})^2+(\frac{p}{1-p})^2\big)^2 \right)^d \E_{D_p}[f(\phi_1, \dots, \phi_n)^2]^2.$$
Note that since $x_i$ can be written as a linear function of $\phi_i$ and linear transformation does not change the degree of the multilinear polynomial, we also have for any degree-at-most $d$ multilinear polynomial $g$ on $x_1,\cdots,x_n$, $$\E_{D_p}[g(x_1, \dots, x_n)^4] \le 12 \cdot d^{3/2} \cdot \left(256 \cdot \big((\frac{1-p}{p})^2+(\frac{p}{1-p})^2\big)^2 \right)^d \E_{D_p}[g(x_1, \dots, x_n)^2]^2.$$
\end{corollary}
\begin{proofof}{Theorem~\ref{hypercontractivity_global_L2}}
We assume the inequality holds for any degree $<d$ polynomials and use induction on the number of characters in a degree $d$ multilinear polynomial $f$ to prove that if the multilinear polynomial $f$ of $\phi_1,\cdots,\phi_n$ depends on at most $k$ characters of $\phi_1,\cdots,\phi_n$, then $\E_{D_p}[f^4] \le d^{3/2} \cdot C^d \cdot \beta^{k} \cdot \|f\|_2^4$ for $C=256 \cdot \left((\frac{1-p}{p})^2+(\frac{p}{1-p})^2\right)^2 $ and $\beta=1+1/n$. 

\paragraph{Base case.} $f$ is a constant function that is independent from $\phi_1,\cdots,\phi_n$, $\E_{D_p}[f^4]=\hf(\emptyset)^4 = \|f\|_2^4.$

\paragraph{Induction step.} Suppose there are $k \ge 1$ characters of $\phi_1,\cdots,\phi_n$ in $f$. Without loss of generality, we assume $\phi_1$ is one of the characters in $f$ and rewrite $f=\phi_1 h_0 + h_1$ for a degree $d-1$ polynomial $h_0$ with at most $k-1$ characters and a degree $d$ polynomial $h_1$ with at most $k-1$ characters. Because $f$ is a multilinear polynomial, $\|f\|_2^2=\|h_0\|_2^2 + \|h_1\|_2^2$. We expand $\E_{D_p}[f^4]=\E_{D_p}[(\phi_1 h_0 + h_1)^4]$ to $$\E_{D_p}[\phi_1^4 \cdot h_0^4] + 4 \E_{D_p}[\phi_1^3 \cdot h_0^3 \cdot h_1] + 6 \E_{D_p}[\phi_1^2 \cdot h_0^2 \cdot h_1^2] + 4 \E_{D_p}[\phi_1 \cdot h_0 \cdot h_1^3]+ \E_{D_p}[h_1^4].$$

From the induction hypothesis, $\E_{D_p}[h_1^4] \le C^d \beta^{k-1} \|h_1\|_2^4$ and 
\begin{equation}\label{first_term}
\E_{D_p}[\phi_1^4 \cdot h_0^4] \le \max\{(\frac{1-p}{p})^2,(\frac{p}{1-p})^2\}\E_{D_p}[h_0^4] \le d^{3/2} \cdot C^{d-0.5} \beta^{k-1} \|h_0\|_2^4.
\end{equation}

Hence $\E_{D_p}[\phi_1^2 \cdot h_0^2 \cdot h_1^2] \le (\E_{D_p}[\phi_1^4 h_0^4])^{1/2}(\E_{D_p}[h_1^4])^{1/2}$ from the Cauchy-Schwarz inequality. From the above discussion, this is at most 
\begin{equation}\label{mid_term}
d^{3/4} \cdot C^{d/2} \cdot \beta^{(k-1)/2} \cdot \|h_1\|_2^2 \cdot d^{3/4} \cdot C^{(d-0.5)/2} \beta^{(k-1)/2}\|h_0\|_2^2 \le d^{3/2} \cdot C^{d-1/4} \cdot \beta^{k-1} \cdot \|h_0\|_2^2 \|h_1\|_2^2.
\end{equation}

Applying the inequality of arithmetic and geometric means on $\E_{D_p}[\phi^3_1 \cdot h_0^3 \cdot h_1]$, we know it is at most 
\begin{equation}\label{second_term}
(\E_{D_p}[\phi_1^4 h_0^4]+\E_{D_p}[\phi_1^2 h_0^2 h_1^2])/2 \le d^{3/2} \big( C^{d-0.5} \beta^{k-1} \|h_0\|_2^4 + C^{d-1/4} \cdot \beta^{k-1} \cdot \|h_0\|_2^2 \|h_1\|_2^2 \big)/2.
\end{equation}

Finally, we bound $\E_{D_p}[\phi_1 \cdot h_0 \cdot h_1^3]$. However, we cannot apply the Cauchy-Schwarz inequality or the inequality of arithmetic and geometric means, because we cannot afford a term like $d^{3/2} \cdot C^d \beta^{k-1} \|h_1\|_2^4$ any more.  We use $D_{\phi_1>0}$ ($D_{\phi_1<0}$ resp.) to denote the conditional distribution of ${D_p}$ on fixing $\phi_1=\sqrt{\frac{p}{1-p}}$ ($-\sqrt{\frac{1-p}{p}}$ resp.) and rewrite 
\begin{equation}\label{fourth_term}
\E_{D_p}[\phi_1 \cdot h_0 \cdot h_1^3]=\sqrt{p(1-p)} \E_{D_{\phi_1>0}}[h_0 \cdot h_1^3] - \sqrt{p(1-p)} \E_{D_{\phi_i<0}}[h_0 \cdot h_1^3].
\end{equation}
Let $L$ be the matrix corresponding to the quadratic form $\E_{D_{\phi_1>0}}[f g] - \E_{D_{\phi_1<0}}[fg]$ for low-degree multilinear polynomials $f$ and $g$ (i.e. let $L$ be a matrix such that $f^T L g = \E_{D_{\phi_1>0}}[fg] - \E_{D_{\phi_1<0}}[fg]$). The main technical lemma of this section is a upper bound on the spectral norm of $L$.

\begin{lemma}\label{small_gap_eigenvalue}
Let $g$ be a degree $d(d \ge 1)$ multilinear polynomial on characters $\phi_2,\cdots,\phi_n$,  we have $$|\E_{D_{\phi_1>0}}[g^2] - \E_{D_{\phi_1<0}}[g^2]| \le \frac{3 d^{3/2}}{p(1-p)} \cdot \frac{\|g\|_2^2}{\sqrt{n}} .$$
Therefore, the spectral norm of $L$ is upper bounded by $\frac{3 d^{3/2}}{p(1-p)} \cdot \frac{\|g\|_2^2}{\sqrt{n}} $.
\end{lemma} 
From the above lemma, we rewrite the equation \eqref{fourth_term} from the upper bound of its eigenvalues:
\begin{multline*}
\E_{D_p}[\phi_1 \cdot h_0 \cdot h_1^3] =  \sqrt{p(1-p)} \left(\E_{D_{\phi_1>0}}[h_0 \cdot h_1^3] - \E_{D_{\phi_1<0}}[h_0 \cdot h_1^3]\right) \\
 =  \sqrt{p(1-p)} (h_0 h_1)^T L \cdot (h_1^2)  
   \le  \sqrt{p(1-p)}  \cdot \frac{3 (2d)^{3/2}}{p(1-p)} \cdot \frac{1}{\sqrt{n}} \cdot \|h_0 h_1\|_2 \cdot \|h_1^2\|_2 .
\end{multline*}
Then we use the inequality of arithmetic and geometric means on it: 
\[\E_{D_p}[\phi_1 \cdot h_0 \cdot h_1^3] \le  \sqrt{p(1-p)} \cdot \frac{10 d^{3/2}}{p(1-p)} \cdot \frac{\|h_0 h_1\|^2_2 + \|h_1^2\|^2_2/n}{2} . \] 
Next, we use the $2 \to 4$ hypercontractivity $\|h^2\|_2^2=\E_{U_p}[h^4] \le 9^d \cdot \left(\frac{p^2}{1-p}+\frac{(1-p)^2}{p}\right)^d \|h\|_2^4$ in $U_p$ and the Cauchy-Schwarz inequality to further simplify it to: 
\begin{multline}\label{new_fourth_term}
\E_{D_p}[\phi_1 \cdot h_0 \cdot h_1^3] \le \frac{5d^{3/2}}{\sqrt{p(1-p)}} \cdot (\|h_0^2\|_2 \cdot \|h_1^2\|_2 + \frac{9^d \cdot \left(\frac{p^2}{1-p}+\frac{(1-p)^2}{p}\right)^d}{n} \|h_1\|^4_2) \\
\le \frac{5d^{3/2}}{\sqrt{p(1-p)}} \cdot 9^d \cdot \left(\frac{p^2}{1-p}+\frac{(1-p)^2}{p}\right)^d(\|h_0\|^2_2 \cdot \|h_1\|^2_2 + \frac{1}{n} \|h_1\|^4_2) .
\end{multline}

From all discussion, we bound $E[f^4]$ by the upper bound of each inequalities in \eqref{first_term},\eqref{second_term},\eqref{mid_term},\eqref{new_fourth_term}: 
\begin{align*}
&E_{D_p}[(\phi_1 h_0 + h_1)^4]\\
=& E_{D_p}[\phi_1^4 \cdot h_0^4] + 4 E_{D_p}[\phi_1^3 \cdot h_0^3 \cdot h_1] + 6E_{D_p}[\phi_1^2 \cdot h_0^2 \cdot h_1^2] + 4 E_{D_p}[\phi_1 \cdot h_0 \cdot h_1^3]+ E_{D_p}[h_1^4]\\
\le & d^{3/2} C^{d-0.5} \beta^{k-1} \|h_0\|_2^4 + 4d^{3/2} \left(C^{d-0.5} \beta^{k-1} \|h_0\|_2^4 + C^{d-1/4} \beta^{k-1} \|h_0\|_2^2 \|h_1\|_2^2\right)/2  + 6 d^{3/2} C^{d-1/4} \beta^{k-1} \|h_0\|_2^2 \|h_1\|_2^2 \\ & \qquad + 4\frac{5d^{3/2}}{\sqrt{p(1-p)}} \cdot 9^d \cdot \left(\frac{p^2}{1-p}+\frac{(1-p)^2}{p}\right)^d \left(\|h_0\|_2^2 \|h_1\|^2_2 + \frac{1}{n} \|h_1\|^4_2\right) + d^{3/2} \cdot C^d \beta^{k-1} \|h_1\|_2^4 \\
\le & 3 d^{3/2} \cdot C^{d-0.5} \beta^{k-1} \|h_0\|_2^4 + \left(8 d^{3/2} \cdot C^{d-1/4} \cdot \beta^{k-1} + \frac{20d^{3/2}}{\sqrt{p(1-p)}} \cdot 9^d \cdot \left(\frac{p^2}{1-p}+\frac{(1-p)^2}{p}\right)^d\right)\cdot \|h_0\|_2^2 \|h_1\|_2^2 \\ & \qquad + \left(\frac{20d^{3/2}}{\sqrt{p(1-p)}} \cdot 9^d \cdot \left(\frac{p^2}{1-p}+\frac{(1-p)^2}{p}\right)^d \cdot \frac{1}{n} + d^{3/2} \cdot C^d \beta^{k-1}\right)\|h_1\|_2^4\\
\le & d^{3/2} \cdot C^d \beta^k \|h_0\|_2^4 + d^{3/2} \cdot C^d \beta^k \cdot 2 \|h_0\|_2^2 \|h_1\|_2^2 + d^{3/2} \cdot (C^d/n + C^d \beta^{k-1}) \|h_1\|_2^4 \le d^{3/2} \cdot C^d \beta^k  \|f\|_2^4.
\end{align*}
\end{proofof}

We prove Lemma \ref{small_gap_eigenvalue} to finish the proof. Intuitively, both $\E_{D_{\phi_1>0}}[g^2]$ and $\E_{D_{\phi_1<0}}[g^2]$ are close to $\E_{D_p}[g^2]$ (we add the dummy character $\phi_1$ back in ${D_p}$) for a low-degree multilinear polynomial $g$; therefore their gap should be small compared to $\E_{D_p}[g^2]=O(\|g\|_2^2)$. Recall that ${D_p}$ is a uniform distribution on the constraint $\sum_i \phi_i=0$, i.e., there are always $n_{+}$ characters of $\phi_i$ with $\sqrt{\frac{p}{1-p}}$ and $n_{-}$ characters with $-\sqrt{\frac{1-p}{p}}$. For convenience, we abuse the notation $\phi$ to denote a vector of characters $(\phi_1,\cdots,\phi_n)$.

\begin{proofof}{Lemma \ref{small_gap_eigenvalue}}
Let $F$ be the $p$-biased distribution on $n-2$ characters with the global cardinality constraint $\sum_{i=2}^{n-1} \phi_i=-\sqrt{\frac{p}{1-p}}+\sqrt{\frac{1-p}{p}}=-q$, i.e.,  $n_{+}-1$ of the characters are always $\sqrt{\frac{p}{1-p}}$ and $n_{-}-1$ of the characters are always $-\sqrt{\frac{1-p}{p}}$. Let $\vec{\phi}_{-i}$ denote the vector $(\phi_2,\cdots,\phi_{i-1},\phi_{i+1},\cdots,\phi_{n})$ of $n-2$ characters such that we could sample $\vec{\phi}_{-i}$ from $F$. Hence $\phi \sim D_{\phi_1<0}$ is equivalent to the distribution that first samples $i$ from ${2,\cdots,n}$ then fixes $\phi_i=\sqrt{\frac{p}{1-p}}$ and samples $\vec{\phi}_{-i} \sim F$. Similarly, $\phi \sim D_{\phi_1>0}$ is equivalent to the distribution that first samples $i$ from ${2,\cdots,n}$ then fixes $\phi_i=-\sqrt{\frac{1-p}{p}}$ and samples $\vec{\phi}_{-i} \sim F$. 

For a multilinear polynomial $g$ depending on characters $\phi_2,\cdots,\phi_n$, we rewrite 
\begin{multline}\label{small_gap_g}
\E_{D_{\phi_1>0}}[g^2] - \E_{D_{\phi_1<0}}[g^2] \\= \E_{i} \E_{\phi \sim F}\left[g(\phi_i=\sqrt{\frac{p}{1-p}},\vec{\phi}_{-i}=\phi)^2 - g(\phi_i=-\sqrt{\frac{1-p}{p}},\vec{\phi}_{-i}=\phi)^2\right].
\end{multline} We will show its eigenvalue is upper bounded by $\frac{3 d^{3/2}}{p(1-p)} \cdot \sqrt{1/n}$. We first use the Cauchy-Schwarz inequality:
\begin{align*}
&\E_{\phi \sim F}\bigg[\left(g(\phi_i=\sqrt{\frac{p}{1-p}}, \vec{\phi}_{-i}=\phi)-g(\phi_i=-\sqrt{\frac{1-p}{p}}, \vec{\phi}_{-i}=\phi)\right)\\
& \qquad\qquad\qquad\qquad\qquad \cdot \left(g(\phi_i=\sqrt{\frac{p}{1-p}}, \vec{\phi}_{-i}=\phi)+g(\phi_i=-\sqrt{\frac{1-p}{p}}, \vec{\phi}_{-i}=\phi)\right)\bigg]\\
\le & \E_{\phi \sim F}\left[\left(g(\phi_i=\sqrt{\frac{p}{1-p}}, \vec{\phi}_{-i}=\phi)-g(\phi_i=-\sqrt{\frac{1-p}{p}}, \vec{\phi}_{-i}=\phi)\right)^2\right]^{1/2}  \\
& \qquad\qquad\qquad\qquad\qquad  \cdot \E_{\phi \sim F}\left[\left(g(\phi_i=\sqrt{\frac{p}{1-p}}, \vec{\phi}_{-i}=\phi) + g(\phi_i=-\sqrt{\frac{1-p}{p}}, \vec{\phi}_{-i}=\phi)\right)^2\right]^{1/2} .
\end{align*}
From the inequality of arithmetic and geometric means and the fact that $\E_{D_p}[g^2]\le d \|g\|_2^2$, observe that the second term is bounded by 
\begin{align*}
& \E_{\phi \sim F}\left[\left(g(\phi_i=\sqrt{\frac{p}{1-p}}, \vec{\phi}_{-i}=\phi) + g(\phi_i=-\sqrt{\frac{1-p}{p}}, \vec{\phi}_{-i}=\phi)\right)^2\right]\\
\le &2\E_{D_{\phi_1>0}}\left[g(\phi_i=\sqrt{\frac{p}{1-p}}, \vec{\phi}_{-i}=\phi)^2\right] + 2\E_{D_{\phi_1<0}}\left[g(\phi_i=-\sqrt{\frac{1-p}{p}}, \vec{\phi}_{-i}=\phi)^2\right] \\
\le & \frac{3}{p(1-p)} \E_{D_p}[g^2] \le \frac{3d}{p(1-p)} \cdot \|g\|^2_2.
\end{align*}

Then we turn to $\E_{\phi \sim F}\bigg[\left(g(\phi_i=\sqrt{\frac{p}{1-p}}, \vec{\phi}_{-i}=\phi) - g(\phi_i=-\sqrt{\frac{1-p}{p}}, \vec{\phi}_{-i}=\phi)\right)^2\bigg]^{1/2}.$ We use\\ $g_i=\sum_{S : i \in S}\hf(S) \phi_{S \setminus i}$ to replace $g(\phi_i=\sqrt{\frac{p}{1-p}}, \vec{\phi}_{-i}=\phi)-g(\phi_i=-\sqrt{\frac{1-p}{p}}, \vec{\phi}_{-i}=\phi)$:
$$\E_{\phi \sim F}\left[\left(\sum_{S: i\in S}\hg(S) \left(\sqrt{\frac{p}{1-p}}+\sqrt{\frac{1-p}{p}}\right)\phi_{S \setminus i}\right)^2\right]^{1/2} = \left(\sqrt{\frac{p}{1-p}}+\sqrt{\frac{1-p}{p}}\right) \E_{\phi \sim F}\left[g_i(\phi)^2\right]^{1/2}.$$

Eventually we bound $\E_{\phi \sim F}[g_i(\phi)^2]$ by its eigenvalue. Observe that $F$ is the distribution on $n-2$ characters with $(n_{+}-1)$ $\sqrt{\frac{p}{1-p}}$'s and $(n_{-}-1)$ $-\sqrt{\frac{1-p}{p}}$'s, which indicates $\sum_j \phi_j + q=0$ in $F$. However, the small difference between $\sum_j \phi_j + q=0$ and $\sum_j \phi_j=0$ will not change the major term in the eigenvalues of ${D_p}$. From the same analysis, the largest eigenvalue of $\E_F[g^2_i]$ is at most $d$. For completeness, we provide a calculation in Section \ref{subsec_bound_eigenvalue_F}.
\begin{claim}\label{eigenvalue_D_exclude2}
For any degree-at-most $d$ multilinear polynomial $g_i$, $\E_{\phi \sim F}[g_i(\phi)^2] \le d\|g_i\|_2^2$.
\end{claim}

Therefore we have $\E_i [\E_{\phi \sim F}[g_i(\phi)^2]^{1/2}]\le \sqrt{d} \cdot \E_i [\|g_i(\phi)\|_2]$ and simplify the right hand side of inequality \eqref{small_gap_g} further:
\begin{align*}
& \E_i\bigg\{ \E_{\phi \sim F}\left[\left(g(\phi_i=\sqrt{\frac{p}{1-p}}, \vec{\phi}_{-i}=\phi)-g(\phi_i=-\sqrt{\frac{1-p}{p}}, \vec{\phi}_{-i}=\phi)\right)^2\right]^{1/2} \\
&\qquad\qquad\qquad\qquad \cdot \E_{\phi \sim F}\left[\left(g(\phi_i=\sqrt{\frac{p}{1-p}}, \vec{\phi}_{-i}=\phi) + g(\phi_i=-\sqrt{\frac{1-p}{p}}, \vec{\phi}_{-i}=\phi)\right)^2\right]^{1/2} \bigg\}\\
&\le \E_i \left[ \left(\sqrt{\frac{p}{1-p}}+\sqrt{\frac{1-p}{p}}\right) \cdot \sqrt{d} \cdot \|g_i\|_2 \cdot \sqrt{\frac{3d}{p(1-p)}} \cdot \|g\|_2\right] \\
&\le   \frac{3 d}{p(1-p)} \cdot \|g\|_2 \cdot \E_i [\|g_i\|_2]
\end{align*}
Using the fact that $\sum_i \|g_i\|_2^2 \le d \|g\|_2^2$ and Cauchy-Schwartz again, we further simplify the expression above to obtain the desired upper bound on the absolute value of eigenvalues of $\E_{D_{\phi_1=1}}[g^2] - \E_{D_{\phi_1=-1}}[g^2]$ :
\begin{multline*}
 \frac{3 d}{p(1-p)} \cdot \|g\|_2 \cdot \E_i [\|g_i\|_2]
\le \frac{3 d}{p(1-p)} \cdot \|g\|_2 \cdot \frac{(\sum_i \|g_i\|^2_2)^{1/2}}{\sqrt{n}} \\ 
\le \frac{3 d}{p(1-p)} \cdot \|g\|_2 \cdot \sqrt{\frac{d}{n}} \|g\|_2 \le \frac{3 d^{3/2}}{p(1-p)} \cdot  \frac{\|g\|^2_2}{\sqrt{n}} .
\end{multline*}
\end{proofof}

\subsection{Proof of Claim \ref{eigenvalue_D_exclude2}}\label{subsec_bound_eigenvalue_F}
We follow the approach in Section \ref{eigenvalues} to determine the eigenvalues of $\E_F[f^2]$ for a polynomial $f=\sum_{S \in {[n-2] \choose \le d}}\hf(S)\phi_S$. Recall that $\sum_i \phi_i + q = 0$ over all support of $F$ for $q=\frac{2p-1}{\sqrt{p(1-p)}}$ as defined before. 

We abuse $\delta_k=\E_F[\phi_S]$ for a subset $S$ with size $k$. We start with $\delta_0=1$ and $\delta_1=-q/(n-2)$. From $(\sum_i \phi_i + q)\phi_S \equiv 0$ for $k\le d-1$ and any $S \in {[n] \choose k}$, we have 
\begin{align*}
\E[\sum_{j \in S} \phi_j \phi_S + q\chi_S + \sum_{j \notin S} \phi_{S \cup j}] = 0 \Rightarrow k \cdot \delta_{k-1} + (k+1)q \cdot \delta_k + (n-2 - k) \cdot \delta_{k+1}=0
\end{align*}
Now we determine the eigenspaces of $\E_F[f^2]$. The eigenspace of an eigenvalue 0 is $span\{(\sum_i \phi + q)\phi_S|S \in {[n-2] \choose \le d-1}\}$. There are $d+1$ non-zero eigenspaces $V_0,V_1,\cdots,V_d$. The eigenspace $V_i$ of $\E_F[f^2]$ is spanned by $\{\hf(S)\phi_S|S \in {[n-2] \choose i}\}$. For each $f \in V_i$, $f$ satisfies the following properties:
\begin{enumerate}
\item $\forall T' \in {[n-2] \choose i-1}, \sum_{j \notin T'}\hf(S \cup j)=0$ (neglect this contraint for $V_0$).
\item $\forall T \in {[n] \choose <i}, \hf(T)=0$.
\item $\forall T \in {[n-2] \choose >i}, \hf(T)=\alpha_{k,|T|}\sum_{S \in T}\hf(S)$ where $\alpha_{k,k}=1$ and $\alpha_{k,k+i}$ satisfies $$i \cdot \alpha_{k,k+i-1} + (k+i+1)\cdot q \cdot\alpha_{k,k+i} + (n-2-2k-i)\alpha_{k,k+i+1}=0.$$
\end{enumerate}
We show the calculation of $\alpha_{k,k+i}$ as follows: fix a subset $T$ of size $k+i$ and consider the orthogonality between $(\sum_i \phi_i + q)]\phi_T$ and $f\in V_k$:
\begin{align*}
& \sum_{j \in T} \alpha_{k,k+i-1} \sum_{S \in {T\setminus j \choose k}} \hf(S) + (k+i+1)\cdot q\cdot \alpha_{k,k+i}\sum_{S \in {T \choose k}}\hf(S) + \sum_{j \notin T} \alpha_{k,k+i+1} \sum_{S \in {T \cup j \choose k}}\hf(S)=0  \\
\Rightarrow & \sum_{S \in {T \choose k}}\Big((k+i+1)\cdot q \cdot\alpha_{k,k+i} + (n-2-k-i)\alpha_{k,k+i+1} + i \cdot \alpha_{k,k+i-1}\Big) \hf(S)\\
&\qquad\qquad\qquad\qquad\qquad\qquad\qquad\qquad\qquad\qquad\qquad+ \sum_{T' \in {T \choose k-1}} \alpha_{k,k+i+1} \sum_{j \notin T} \hf(T' \cup j)=0 .
\end{align*}
Using the first property $\forall T' \in {[n-2] \choose i-1}, \sum_{j \notin T'}\hf(S \cup j)=0$ to remove all $S' \notin T$, we have 
$$\sum_{S \in {T \choose k}}\big(i \cdot \alpha_{k,k+i-1} + (k+i+1)\cdot q \cdot\alpha_{k,k+i} + (n-2-2k-i)\alpha_{k,k+i+1} \big)\hf(S)=0$$
We calculate the eigenvalues of $V_k$ following the approach in Section \ref{eigenvalues}. Fix $S$ and $S'$ with $i = |S \vartriangle S'|$, we still use $\tau_i$ to denote the coefficients of $\hf(S')$ in the expansion of $\sum_{T} (\delta_{S \vartriangle T}-\delta_S \delta_T)\hf(T)$. Observe that $\tau_i$ is as same as the definition in Section \ref{eigenvalues} in terms of $\delta$ and $\alpha_k:$
\begin{align*}
\tau_0&=\sum_{i=0}^{d-k} {n-2-k \choose i} \cdot \alpha_{k,k+i} \cdot \delta_i, & \tau_{2l}&=\sum_{i=0}^{d-k}\alpha_{k,k+i} \sum_{t=0}^{i}{l \choose t}{n-2-k-2l \choose i-t}\delta_{2l+i-2t} .
\end{align*}

Observe that the small difference between $(\sum_i \phi_i + q)\equiv 0$ and $\sum_i \phi_i\equiv 0$ only changes a little in the recurrence formulas of $\delta$ and $\alpha$. For $\delta_{2i}$ and $\alpha_{k,k+2i}$ of an integer $i$, the major term is still determined by $\delta_{2i-2}$ and $\alpha_{k,k+2i-2}$. For $\delta_{2i+1}$ and $\alpha_{k,k+2i+1}$, they are still in the same order (the constant before $n^{-i-1}$ will not change the order). Using the same induction on $\delta$ and $\alpha$, we have
\begin{enumerate}
\item $\delta_{2i}=(-1)^i \frac{(2i-1)!!}{n^i}+O(n^{-i-1})$;
\item $\delta_{2i+1}=O(n^{-i-1})$;
\item $\alpha_{k,k+2i}=(-1)^i \frac{(2i-1)!!}{n^i}+O(n^{-i-1})$;
\item $\alpha_{k,k+2i+1}=O(n^{-i-1})$.
\end{enumerate}
Hence $\tau_0=\sum_{even \ i=0}^{d-k} \frac{(i-1)!!(i-1)!!}{i!} + O(n^{-1})$ and $\tau_{2l}=O(n^{-l}).$ Follow the same analysis in Section \ref{eigenvalues}, we know the eigenvalue of $V_k$ is $\tau_0 \pm O(\tau_{2})=\sum_{even \ i=0}^{d-k} \frac{(i-1)!!(i-1)!!}{i!} + O(n^{-1})$.

From all discussion above, the eigenvalues of $\E_F[f^2]$ is at most $[\frac{d}{2}]+1 \le d$.

\section{Parameterized algorithm for CSPs above average with global cardinality constraints}\label{sec_FPT_global}
We show that CSPs above average with the global cardinality constraint $\sum_i x_i=(1-2p)n$ are fixed-parameter tractable for any $p \in [p_0,1-p_0]$ with an integer $pn$. We still use $D_p$ to denote the uniform distribution on all assignments in $\{\pm 1\}^n$ complying with $\sum_i x_i=(1-2p)n$.

Without loss of generality, we assume $p < 1/2$ and $(1-2p)n$ is an integer. We choose the standard basis $\{\chi_S\}$ in this section instead of $\{\phi_S\}$, because the Fourier coefficients in $\{\phi_S\}$ can be arbitrary small for some $p \in (0,1)$. 
\begin{theorem}\label{CSP_global}
For any constant $p_0 \in (0,1)$ and $d$, given an instance of $d$-ary CSP, a parameter $t$, and a parameter $p \in [p_0,1-p_0]$, there exists an algorithm with running time $n^{O(d)}$ that either finds a kernel on at most $C \cdot t^2$ variables or certifies that $OPT \ge AVG+t$ under the global cardinality constraint $\sum_i x_i=(1-2p)n$ for a constant $C=160 d^2 \cdot 30^d \left((\frac{1-p_0}{p_0})^2+(\frac{p_0}{1-p_0})^2\right)^d \cdot (d!)^{3d^2} \cdot (1/2p_0)^{4d}$.
\end{theorem}

\subsection{Rounding}

Let $f$ be a degree $d$ polynomial whose coefficients are multiples of $\gamma$ in the standard basis $\{\chi_S|S \in {[n] \choose \le d}\}$. We show how to find an integral-coefficients polynomial $h$ such that $f - \big(\sum_i x_i - (1-2p) n\big)h$ only depends on $O(\Var_{D_p}(f))$ variables. We use the rounding algorithm in Section \ref{rounding_bisection} as a black box, which provides a polynomial $h$ such that $f-\big(\sum_i x_i \big)h$ only depends on $O(\Var_D(f))$ variables (where $D$ is the distribution conditioned on the bisection constraint). Without loss of generality, we assume $\hf(\emptyset)=0$ because $\Var_{D_p}(f)$ is independent with $\hf(\emptyset)$. 

Before proving that $f$ depends on at most $O(\Var_{D_p}(f))$ variables, we first define the inactivity of a variable $x_i$ in $f$.
\begin{definition}
A variable $x_i$ for $i \in [n]$ is inactive in $f=\sum_S \hf(S) \chi_S$ if $\hf(S)=0$ for all $S$ containing $x_i$. \\ A variable $x_i$ is inactive in $f$ under the global cardinality constraint $\sum_i x_i = (1-2p) n$ if there exists a polynomial $h$ such that $x_i$ is inactive in $f - \big(\sum_i x_i - (1-2p) n\big)h$. 
\end{definition}
In general, there are multiple ways to choose $h$ to turn a variable into inactive. However, if we know a subset $S$ of $d$ variables and the existence of some $h$ to turn $S$ into inactive in $f - \big(\sum_i x_i - (1-2p) n\big)h$, we show that $h$ is uniquely determined by $S$. Intuitively, for any subset $S_1$ with $d-1$ variables, there are ${S_1 \cup S \choose d}={2d-1 \choose d}$ ways to choose a subset of size $d$. Any $d$-subset $T$ in $S_1 \cup S$ contains at least one inactive variable such that $\hf(T) - \sum_{j \in T} \hh(T \setminus j)=0$ from the assumption. At the same time, there are at most ${2d-1 \choose d-1}$ coefficients of $\hh$ in $S_1 \cup S$. So there is only one solution of coefficients in $\hh$ to satisfy these ${2d-1 \choose d}$ equations. 
\begin{claim}\label{unique_h}
Given $f=\sum_{T \in {[n] \choose \le d}}\hf(T)\chi_T$ and $p \in [p_0,1-p_0]$, let $S$ be a subset with at least $d$ variables such that there exists a degree $\le d-1$ multilinear polynomial $h$ turning $S$ into inactive in $f - \big(\sum_i x_i - (1-2p) n\big)h$. Then $h$ is uniquely determined by any $d$ elements in $S$ and $f$.
\end{claim}
\begin{proof}
Without lose of generality, we assume $S=\{1,\cdots,d\}$ and determine $\hh(S_1)$ for $S_1=\{i_1,\cdots,i_{d-1}\}$. 

For simplicity, we first consider the case $S \cap S_1=\emptyset$. From the definition, we know that for any $T \in {S \cup S_1 \choose d}$, $T$ contains at least one inactive variable, which indicates $\hf(T) - \sum_{j \in T} \hh(T \setminus j)=0.$ Hence we can repeat the argument in Lemma \ref{eqn_d} to determine $\hh(S_1)$ from $\hf(T)$ over all $T \in {S \cup S_1 \choose d}$.

Let $\beta_{d-1,1 }=(d-2)!$ and $\beta_{d-i-1, i+1}=\frac{-i}{d-i-1}\beta_{d-i,i}$ for any $i \in \{1,\cdots,d-2\}$ be the parameters define in Lemma \ref{eqn_d}. For any $S_2 \in {S \choose d-1}$, by the same calculation, $$\sum_{i=1}^{d-1}\beta_{d-i,i} \sum_{T_1 \in {S_1 \choose d-i}, T_2 \in {S_2 \choose i}} \left(\hf(T_1 \cup T_2) - \sum_{j \in T_1 \cup T_2} \hh(T_1 \cup T_2 \setminus j)\right)=0$$ indicates that (all $\hh(S)$ not $S_1$ or $S_2$ cancel with each other) $$(d-1)! \cdot \hh(S_1) + (-1)^d (d-1)! \cdot \hh(S_2) = \sum_{i=1}^{d-1}\beta_{d-i,i} \sum_{T_1 \in {S_1 \choose d-i}, T_2 \in {S_2 \choose i}} \hf(T_1 \cup T_2).$$
Hence $\hh(S_2)=(-1)^{d-1} \hh(S_1) + (-1)^{d} \sum_{i=1}^{d-1} \frac{\beta_{d-i,i}}{(d-1)!} \big(\sum_{T_1 \in {S_1 \choose d-i}, T_2 \in {S_2 \choose i}} \hf(T_1 \cup T_2)\big)$ for any $S_2 \in {S \choose d-1}$. Replacing all $\hh(S_2)$ in the equation $\hf(S) - \sum_{S_2 \in {S \choose d-1}}\hh(S_2)=0$, we obtain $\hh(S_1)$ in terms of $f$ and $S$. 

If $S \cap S_1 \neq \emptyset$, then we set $S'=S \setminus S_1$. Next we add arbitrary $|S \cap S_1|$ more variables into $S'$ such that $|S'|=d$ and $S' \cap S_1 =\emptyset$. Observe that $S' \cup S_1$ contains at least $d$ inactive variables. Repeat the above argument, we could determine $\hh(S_1)$.

After determining $\hh(S_1)$ for all $S_1 \in {[n] \choose d-1}$,  we repeat this argument for $S_1 \in {[n] \choose d-2}$ and so on. Therefore we could determine $\hh(S_1)$ for all $S_1 \in {[n] \choose \le d-1}$ from $S$ and the coefficients in $f$.
\end{proof}
\begin{remark}
The coefficients of $h$ are multiples of $\gamma/d!$ if the coefficients of $f$ are multiples of $\gamma$.
\end{remark}
Let $h_1$ and $h_2$ be two polynomials such that at least $d$ variables are inactive in both $f - \big(\sum_i x_i - (1-2p) n\big)h_1$ and $f - \big(\sum_i x_i - (1-2p) n\big)h_2$. We know that $h_1=h_2$ from the above claim. Furthermore, it implies that any variable that is inactive in $f - \big(\sum_i x_i - (1-2p) n\big)h_1$ is inactive in $f - \big(\sum_i x_i - (1-2p) n\big)h_2$ from the definition, and vice versa.

Based on this observation, we show how to find a degree $d-1$ function $h$ such that there are fews active variables left in $f - \big(\sum_i x_i - (1-2p) n\big)h$. The high level is to random sample a subset $Q$ of $(1-2p)n$ variables and restrict all variables in $Q$ to 1. Thus the rest variables constitutes the bisection constraint on $2pn$ variables such that we could use the rounding process in Section~\ref{rounding_bisection}. Let $k$ be a large number, $Q_1,\cdots,Q_k$ be $k$ random subsets and $h_1,\cdots,h_k$ be the $k$ functions after rounding in Section~\ref{rounding_bisection}. Intuitively, the number of active variables in $f-(\sum_{i \notin Q_1} x_i)h_1,\cdots,f-(\sum_{i \notin Q_i} x_k)h_k$ are small with high probability such that $h_1,\cdots,h_k$ share at least $d$ inactive variables. We can use one function $h$ to represent $h_1,\cdots,h_k$ from the above claim such that the union of inactive variables in $f-(\sum_{i \notin Q_j} x_i)h_j$ over all $j \in [k]$ are inactive in $f - (\sum_i x_i)h$ from the definition. Therefore there are a few active variables in $f - (\sum_i x_i)h$.

Let us move to $f - \big(\sum_i x_i - (1-2p) n\big)h$. Because $h$ is a degree-at-most $d-1$ function, $(1-2p)n \cdot h$ is a degree $\le d-1$ function. Thus we know that the number of active variables among degree $d$ terms in $f - \big(\sum_i x_i - (1-2p) n\big)h$ is upper bounded by the number of active variables in $f - (\sum_i x_i)h$. For the degree $<d$ terms left in $f - \big(\sum_i x_i - (1-2p) n\big)h$, we repeat the above process again.
\begin{theorem}\label{rounding_global}
Given  a global cardinality constraint $\sum_i x_i = (1-2p)n$ and a degree $d$ function $f=\sum_{S \in {[n] \choose \le d}}\hf(S)$ with $\Var_{D_p}(f)<n^{0.5}$ and coefficients of multiples of $\gamma$ , there is an efficient algorithm running in time $O(d n^{2d})$ to find a polynomial $h$ such that there are at most $\frac{C'_{p,d} \cdot \Var_{D_p}(f) }{\gamma^2  }$ active variables in $f - \big(\sum_i x_i - (1-2p) n\big)h$ for $C'_{p,d}=\frac{20 d^2 7^d \cdot (d!)^{2d^2}}{(2p)^{4d}}$. 
\end{theorem}
\begin{proof}
For any subset $Q \in {[n] \choose (1-2p) n}$, we consider the assignments conditioned on $x_Q=\vec{1}$ and use $f_Q$ to denote the restricted function $f$ on $x_Q=\vec{1}$. Conditioned on $x_Q=\vec{1}$, the global cardinality constraint on the rest variables is $\sum_{i \notin Q} x_i=0$. We use $D_Q$ denote the distribution on assignments of $\{x_i|i \notin Q\}$ satisfying $\sum_{i \notin Q} x_i=0$, i.e., the distribution of $\{x_i|i \notin \bar{Q}\}$ under the bisection constraint.

Let $X_Q(i) \in \{0,1\}$ denote whether $x_i$ is active in $f_Q$ under the bisection constraint of $\bar{Q}$ or not after the bisection rounding in Theorem \ref{rounding_h_dCSP}. From Theorem \ref{rounding_h_dCSP}, we get an upper bound on the number of active variables in $f_Q$, i.e., $$\sum_{i}X_Q(i)\le 2 C'_d \cdot \Var_{D_Q}(f_Q)$$ for $C'_d= \frac{7^d (d! (d-1)! \cdots 2!)^2}{\gamma^2}$ and any $Q$ with $\Var_{D_Q}(f_Q)=O(n^{0.6})$.

We claim that $$\E_Q[\Var_{D_Q}(f_Q)]\le \Var_{D_p}(f).$$ From the definition, $\E_Q[\Var_{D_Q}(f_Q)]=\E_Q \E_{y \sim D_Q} [f_Q(y)^2] - \E_Q \big[\E_{y \sim D_Q} [f_Q(y)]\big]^2$. At the same time, we observe that $\E_Q \E_{y \sim D_Q} [f_Q(y)^2]= \E_D [f^2]$ and $\E_Q [\E_{y \sim D_Q} f_Q(y)]^2 \ge \E_{D_p}[f]^2$. Therefore $\E_Q[\Var_{D_Q}(f_Q)] \le \Var_{D_p}[f]$. One observation is that $\Pr_Q[\Var_{D_Q}\ge n^{0.6}]<n^{-0.1}$ from the assumption $\Var_{D_p}(f) <n^{0.5}$, which is very small such that we can neglect it in the rest of proof. From the discussion above, we have $\E_Q[\sum_{i}X_Q(i)] \le 2C'_d \cdot \Var_{D_p}(f)$ with high probability.

Now we consider the number of $i$'s with $\E_Q[X_Q(i)] \le \frac{(2p)^{2d}}{5d}$. Without loss of generality, we use $m$ to denote the number of $i$'s with $\E_Q[X_Q(i)] \le \frac{(2p)^{2d}}{5d}$ and further assume these variables are $\{1,2,\cdots,m\}$ for convenience. Hence for any $i>m$, $\E_Q[X_Q(i)] > \frac{(2p)^{2d}}{5d}$. We know the probability $\Var_{D_Q}(f)\le \frac{2 \Var_{D_p}(f)}{(2p)^{2d}}$ is at least $1 - \frac{(2p)^{2d}}{2}$, which implies $$n-m \le \frac{2C'_d \cdot \Var_{D_Q}(f)}{\frac{(2p)^{2d}}{5d}}\le \frac{20d \cdot C'_d \Var_{D_p(f)}}{(2p)^{4d}}.$$ We are going to show that $\E_Q[X_Q(i)]$ is either 0 or at least $\frac{(2p)^{2d}}{5d}$, which means that only $x_{m+1},\cdots,x_n$ are active in $f$ under $\sum_i x_i = 0$. Then we discuss how to find out a polynomial $h_d$ such that $x_1,\cdots,x_m$ are inactive in the degree $d$ terms of $f- \big(\sum_i x_i - (1-2p) n\big)h_d$.

We fix $d$ variables $x_1,\cdots,x_d$ and pick $d$ arbitrary variables $x_{j_1},\cdots,x_{j_d}$ from $\{d+1,\cdots,n\}$. We focus on $\{x_1,x_2,\cdots,x_d,x_{j_1},\cdots,x_{j_d}\}$ now. With probability at least $(2p)^{2d}-o(1) \ge 0.99 (2p)^{2d}$ over random sampling $Q$, none of these $2d$ variables is in $Q$. At the same time, with probability at least $1-2d \cdot  \frac{(2p)^{2d}}{5d}$, all variables in the intersction $\{x_1,\dots,x_m\}\cap \{x_1,x_2,\cdots,x_d,x_{j_1},\cdots,x_{j_d}\}$ are inactive in $f_Q$ under the bisection constraint on $\bar{Q}$ ($2d$ is for $x_{j_1},\cdots,x_{j_d}$ if necessary). Therefore, with probability at least $0.99 (2p)^{2d}- 2d \cdot  \frac{(2p)^{2d}}{5d} - \frac{(2p)^{2d}}{2}\ge 0.09 (2p)^{2d}$, $x_1,x_2,\cdots,x_d$ are inactive in $f_Q$ under $\sum_{i \notin Q}x_i=0$ and $n-m$ is small. Namely there exists a polynomial $h_{x_{j_1},\cdots,x_{j_d}}$ such that the variables in $\{x_1,\dots,x_m\}\cap \{x_1,x_2,\cdots,x_d,x_{j_1},\cdots,x_{j_d}\}$ are inactive in $f_Q - (\sum_{i \notin Q}x_i)h_{x_{j_1},\cdots,x_{j_d}}$. 

Now we apply Claim \ref{unique_h} on $S=\{1,\cdots,d\}$ in $f$ to obtain the unique polynomial $h_d$, which is the combination of $h_{x_{j_1},\cdots,x_{j_d}}$ over all choices of $j_1,\cdots,j_d$, and consider $f-(\sum_{i} x_i)h_d$. Because of the arbitrary choices of $x_{j_1},\cdots,x_{j_d}$, it implies that $x_1,\cdots,x_d$ are inactive in $f-(\sum_i x_i)h_d$. For example, we fix any $j_1,\cdots,j_d$ and $T=\{1,j_1,\cdots,j_{d-1}\}$. we know $\hf(T)-\sum_{j \in T}\hh_{x_{j_1},\cdots,x_{j_d}}(T \setminus j)=0$ from the definition of $h_{x_{j_1},\cdots,x_{j_d}}$. Because $h_d$ agrees with $h_{x_{j_1},\cdots,x_{j_d}}$ on the Fourier coefficients from Claim \ref{unique_h}, we have $\hf(T)-\sum_{j \in T}\hh_{d}(T \setminus j)=0$.

Furthermore, it implies that $x_{d+1},\cdots,x_m$ are also inactive in $f-(\sum_i x_i)h_d$. For example, we fix $j_1 \in \{d+1,\cdots,m\}$ and choose $j_2,\cdots,j_d$ arbitrarily. Then $x_{1},\cdots,x_{d},$ and $x_{j_1}$ are inactive in $f_Q - (\sum_{i \notin Q} x_i)h_{x_{j_1},\cdots,x_{j_d}}$ for some $Q$ from the discussion above, which indicates that $x_{j_1}$ are inactive in $f - (\sum_i x_i)h_d$ by Claim \ref{unique_h}.  

To find $h_d$ in time $O(n^{2d})$, we enumerate all possible choices of $d$ variables in $[n]$ as $S$. Then we apply Claim \ref{unique_h} to find the polynomial $h_S$ corresponding to $S$ and check $f - (\sum_i x_i)h_S$. If there are more than $m$ inactive variables in $f - (\sum_i x_i)h_S$, then we set $h_d=h_S$. Therefore the running time of this process is ${n \choose d} \cdot O(n^d)=O(n^{2d})$.

Hence, we can find a polynomial $h_d$ efficiently such that at least $m$ variables are inactive in $f - (\sum_i x_i)h_d$. Let us return to the original global cardinality constraint $\sum_i x_i = (1-2p) n$. Let $$f_d=f - \big(\sum_i x_i - (1-2p) n\big)h_d.$$ $x_1,\cdots,x_m$ are no longer inactive in $f_d$ because of the extra term $(1-2p) n \cdot h$. However, $x_1,\cdots,x_m$ are at least independent with the degree $d$ terms in $f_d$. Let $A_d$ denote the set for active variables in the degree $d$ terms of $f_d$, which is less than $\frac{20d \cdot C'_d  \Var_{D_p(f)}}{(2p)^{4d}}$ from the upper bound of $n-m$. 

For $f_d$, observe that $\Var_{D_p}(f_d)=\Var_{D_p}(f)$ and all coefficients of $f_d$ are multiples of $\gamma_{d-1}=\gamma/d!$ from Claim~\ref{unique_h}. For $f_d$, we neglect its degree $d$ terms in $A_d$ and treat it as a degree $d-1$ function from now on. Then we could repeat the above process again for the degree $d-1$ terms in $f_d$ to obtain a degree $d-2$ polynomial $h_{d-1}$ such that the active set $A_{d-1}$ in the degree $d-1$ terms of $f_{d-1}=f_{d} - \big(\sum_i x_i - (1-2p) n\big)h_{d-1}$ contains at most $\frac{20d \cdot C'_{d-1} \Var_{D_p(f)}}{(2p)^{4d}}$ variables for $C'_{d-1}=\frac{((d-1)^2 \cdots 2!)^2}{\gamma_{d-1}^2}$. At the same time, observe that the degree of $\big(\sum_i x_i - (1-2p) n\big)h_{d-1}$ is at most $d-1$ such that it will not introduce degree $d$ terms to $f_{d-1}$. Then we repeat it again for terms of degree $d-2, d-3,$ and so on.

To summarize, we can find a polynomial $h$ such that $A_d \cup A_{d-1} \cdots \cup A_1$ is the active set in $f - \big(\sum_i x_i - (1-2p) n\big)h$. At the same time, $|A_d \cup A_{d-1} \cdots \cup A_1|\le \sum_i |A_i| \le \frac{20 d^2 7^d \cdot \Var_{D_p}(f) \cdot (d!)^{2d^2}}{\gamma^2 \cdot (2p)^{4d}}$.
\end{proof}

\subsection{Proof of Theoreom \ref{CSP_global}}
In this section, we prove Theorem \ref{CSP_global}. Let $f=f_{\cal I}$ be the degree $d$ function associated with the instance $\cal I$ and $g=f-\E_{D_p}[f]$ for convenience. We discuss $\Var_{D_p}[f]$ in two cases.

If $\Var_{D_p}[f]=\E_{D_p}[g^2] \ge 8 \left(16 \cdot \big((\frac{1-p}{p})^2+(\frac{p}{1-p})^2\big) \cdot d^{3}  \right)^d \cdot t^2$, we have $$\E_{D_p}[g^4] \le 12 \cdot \left(256 \cdot \big((\frac{1-p}{p})^2+(\frac{p}{1-p})^2\big)^2 \cdot d^{6}  \right)^d \E_{D_p}[g^2]^2$$ from the $2\to4$ hypercontractivity in Theorem \ref{hypercontractivity_global_2nd}. By Lemma \ref{4th_moment_method}, we know $$\Pr_{D_p}\left[g \ge \frac{\sqrt{\E_{D_p}[g^2]}}{2 \sqrt{12 \cdot \left(256 \cdot \big((\frac{1-p}{p})^2+(\frac{p}{1-p})^2\big)^2 \cdot d^{6}  \right)^d }}\right]>0.$$ Thus $\Pr_{D_p}[g \ge t]>0$, which demonstrates that $\Pr_{D_p}[f \ge \E_{D_p}[g]+t]>0$.

Otherwise we know $\Var_{D_p}[f]\le 8 \left(16 \cdot \big((\frac{1-p}{p})^2+(\frac{p}{1-p})^2\big) \cdot d^{3}  \right)^d \cdot t^2$. We set $\gamma=2^{-d}$. From Theorem \ref{rounding_global}, we could find a degree $d-1$ function $h$ in time $O(n^{2d})$ such that $f - \big(\sum_i x_i - (1-2p)n\big)h$ contains at most $\frac{ C'_{p,d} \cdot \Var_{D_p}(f)}{\gamma^2}$ variables. We further observe that $f(\alpha)=f(\alpha) - \big(\sum_i \alpha_i - (1-2p)n\big)h(\alpha)$ for any $\alpha$ in the support of $D_p$. Then we know the kernel of $f$ and $\cal I$ is at most 
\begin{multline*}
8 \left(16 \cdot \big((\frac{1-p}{p})^2+(\frac{p}{1-p})^2\big) \cdot d^{3}  \right)^d \cdot t^2 \cdot \frac{C'_{p,d}}{\gamma^2}\\< 8 \left(16 \cdot \big((\frac{1-p}{p})^2+(\frac{p}{1-p})^2\big) \cdot d^{3}  \right)^d  \cdot t^2 \cdot \frac{20 d^2 7^d \cdot \Var_{D_p}(f) \cdot (d!)^{2d^2} \cdot 2^{2d}}{\gamma^2 \cdot (2p)^{4d}} <C \cdot t^2.
\end{multline*}

The running time of this algorithm is $O(d n^{2d})$.

\section*{Acknowledgement}
We would like to thank Ryan O'Donnell and Yu Zhao for useful discussion on the hypercontractive inequalities. The first author is grateful to David Zuckerman for his constant support and encouragement, as well as for many fruitful discussions.

\bibliographystyle{plain}
\bibliography{FPT}

\end{document}